\documentclass{elsarticle}
\usepackage{graphicx}
\usepackage{amsmath}
\usepackage{amssymb}
\usepackage{xcolor}

\newcommand{\Sets}{\mathcal{S}}
\newcommand{\IEntropy}{\mathcal{I}}
\newcommand{\DEntropy}{\Delta}
\newcommand{\U}{\mathcal{U}}
\newcommand{\T}{\mathcal{T}}

\newcommand{\travis}[1]{\textcolor{black}{#1}}
\newcommand{\anon}[1]{#1}
\newcommand{\supp}{\mathrm{supp}}
\renewcommand{\log}{\lg}

\newtheorem{theorem}{Theorem}
\newtheorem{lemma}{Lemma}
\newtheorem{corollary}{Corollary}
\newtheorem{definition}{Definition}
\newenvironment{proof}{\begin{trivlist} \item[] {\it Proof.}}{\hfill $\Box$\end{trivlist}}

\begin{document}

\begin{frontmatter}

\title{Compressed Set Representations \\ based on Symmetric Set Difference}
\tnotetext[latin]{An early version of this article appeared in {\em Proc. LATIN 2026}.}

\author[CeBiB]{Travis Gagie\fnref{nserc1,cebib}}
\ead{travis.gagie@gmail.com}
\address[CeBiB]{Center for Biotechnology and Bioengineering (CeBiB), Chile}

\author[Dalhousie]{Meng He\fnref{nserc2}}
\ead{mhe@cs.dal.ca}
\address[Dalhousie]{Dalhousie University, Canada}

\author[DCC,CeBiB]{Gonzalo Navarro\fnref{fondecyt,cebib}}
\ead{gnavarro@dcc.uchile.cl}
\address[DCC]{Department of Computer Science, University of Chile, Chile.}

\fntext[nserc1]{Supported by NSERC Discovery Grant RGPIN-07185-2020, Canada}
\fntext[cebib]{Partially funded by Basal Funds FB0001 and AFB240001, ANID, Chile.}
\fntext[nserc2]{Supported by NSERC Discovery Grant RGPIN-05953-2024, Canada}
\fntext[fondecyt]{Partially funded by Fondecyt Grants 1230755 and 1260080, ANID, Chile.}

\begin{abstract}
We introduce a compressed representation of sets of sets that exploits \textcolor{black}{how similar they are to each other}. Our representation supports access, membership, predecessor and successor queries on the sets within logarithmic time. \textcolor{black}{It also computes the support of any element (i.e., the sets containing it) in optimal time.} In addition, we give a new MST-based construction algorithm for the representation that outperforms standard ones.    
\end{abstract}

\end{frontmatter}

\section{Introduction}

\noindent
The goal of compact data structures is to represent combinatorial objects in space close to their compressed form, so that the representation can efficiently answer a desired set of queries on the objects without the need to decompress them~\cite{Nav16}. Given the incomputability of Kolmogorov's absolute notion of compressibility, the notion of ``compressed form'' varies depending on the application and on the kind of regularities one expects to exploit from the data.

In this paper we focus on representing a set $\Sets$ of finite sets, each $S \in \Sets$ drawn from a totally ordered universe $\U$ of elements, and what we aim to exploit is the fact that some sets may be similar to others, that is, their symmetric difference, $S \triangle S'$, may be small. In such a case, we can represent $S$ in terms of $S'$ in space proportional to $|S \triangle S'|$, by listing which elements we must insert into or delete from $S'$ in order to obtain $S$ (symmetrically, we can represent $S'$ from $S$). 

Let $\DEntropy(\Sets)$ be the size of a representation of $\Sets$ based on encoding symmetric differences.  In this paper we show that a representation of $\Sets$ in space $O(\DEntropy(\Sets))$ can provide efficient general access to the compressed data. Concretely, we focus on providing the following functionality:
\begin{description}
    \item[Membership:] Determine if some element $x \in \U$ belongs to some set $S \in \Sets$.
    \item[Access:] Obtain the $i$th smallest element of some set $S \in \Sets$, for any $1 \le i \le |S|$.
    \item[Rank:] Count the number of elements $\le x$ that exist in $S \in \Sets$, for some $x \in \U$.
    \item[Predecessor and successor:] Find the largest element $\le x$ and the smallest element $\ge x$ in a set $S \in \Sets$.
\end{description}

Those fundamental queries enable many operations on sets of sets, such as for example union and intersection of sets in $\Sets$. Our representation supports the fundamental queries in time $O(\log |\U|)$ or less, which is a low overhead anyway incurred in many cases, even with explicit (uncompressed) representations. 

Sets of sets arise in many applications.
For example, the rows of a Boolean matrix can be viewed as the characteristic vectors of sets, in which case our membership query corresponds to accessing individual matrix cells, while access (as well as predecessor/successor) corresponds to collecting the 1s in a row.  Since a graph can be represented as its Boolean adjacency matrix, we can also view it as a set of sets, in which case membership corresponds to asking for the existence of individual edges and access to traversing the neighbors of a node. Our structures 
can be used to run a variety of algorithms, such as matrix multiplication and graph traversals, while they are represented within space $O(\DEntropy(\Sets))$. This idea has a rich history of developments and applications, which we defer to Section~\ref{sec:motivation}.

A close predecessor of our work is the so-called ``containment entropy'' \cite{ABGNP25}, which represents sets as subsets of others: if $S \subset S'$, then they can represent $S$ by telling which elements of $S'$ it contains. They define a compressibility measure corresponding to the best such representation, and show that within that space they can support the five queries mentioned above on each set $S$ in time $O(\log(|\U|/|S|))$. This can be called ``deletion compressibility'': it specifies what to delete from $S'$ to obtain $S$. Our measure $\DEntropy(\Sets)$ is coarser, in the sense that it measures the number of elements and not the exact number of bits to represent them, but it is more general because it allows insertions and deletions.

We start by defining the simpler concept of ``insertion compressibility'', $\IEntropy(\Sets)$, where we describe $S'$ via the $|S' \setminus S|$ elements we must add to $S$ to obtain $S'$. This is the dual of deletion compressibility: in terms of number of elements, $|S' \setminus S|$, they differ only because the deletion compressibility starts from the set of all elements, $\U$, whereas the insertion compressibility starts from the set $\emptyset$. 

We then define the more powerful ``symmetric-difference compressibility'' (or ``symdiff compressibility'' for short), $\DEntropy(\Sets) \le \IEntropy(\Sets)$, as the size of the minimum arrangement where sets can be defined in terms of others by specifying insertions and deletions to make, starting from basic sets $\emptyset$ and $\U$. We solve the five basic queries within time $O(\log|\U|)$, or less, and space $O(\DEntropy(\Sets))$.  

The five queries, both on the structures of space bounded by $O(\IEntropy(\Sets))$ and $O(\DEntropy(\Sets))$, are efficiently supported on top of the so-called ``tree extraction'' framework \cite{HMZ16,HMZ14}, relatively directly in the case of insertion compressibility, and requiring new ideas in the case of symdiff. In particular, we show how to access the $i$th smallest element among the ancestors of some node, when some nodes in the path can ``delete'' symbols that exist upward in the path. \textcolor{black}{As discussed in Section~\ref{sec:ops}, this generalizes the segment rank and segment select problem~\cite{bgpt25} to trees, and thus, it might have independent interest.}

We show that insertion compressibility, and the more challenging symdiff compressibility, can be computed efficiently. As a byproduct of the latter, we give improved results for the MST-based construction of the compressed matrix representation of Alves et al.~\cite{AMB+25}, which \textcolor{black}{enables multiplying the matrix} by a dense vector in time $O(\DEntropy(\Sets))$. This is a crucial aspect to make this representation usable, up to the point that previous works \cite{BL01} preferred to use approximate MSTs. We also show how to compute insertion compressibility, $\IEntropy(\Sets)$, which can be used directly to improve the time to compute the containment entropy \cite{ABGNP25}.

\textcolor{black}{We also study an operation that relates the sets in $\Sets$: the ``support'' of a universe element, that is, the sets that contain it.
\begin{description}
\item[Support:] Given $x \in \U$, return the set $\supp(x,\Sets) = \{ S \in \Sets, x \in S \}$.
\end{description}
We show how $\supp(x,\Sets)$ can be computed (and even enumerated) in optimal time, or its size be computed in constant time, all within $O(\Delta(\Sets))$ space. This is equivalent to retrieving the 1s in a column of the matrix described above, but we show that we cannot exploit this symmetry by simply also storing the transposed matrix, without possibly asymptotically worsening its symdiff compressibility.}

\textcolor{black}{We end by exploring possible extensions to our symdiff compressibility model, but showing that either they can be accomodated within $O(\Delta(\Sets))$ space, or they become NP-hard to build optimally or approximate within constant factors.}

\textcolor{black}{A preliminary version of this article was presented in {\em Proc. LATIN'26} \cite{GHNlatin24}. This version includes succinct versions of the data structures, the computation of the support, the relation with the segment rank and select problem, the study of stronger compressibility models, and a polished writeup.}

\section{Preliminaries and the Tree Extraction Framework} \label{sec:tree-extraction}

\noindent
Throughout this paper, we adopt the word RAM model with $\omega$-bit words.
As we heavily use the tree extraction framework \cite{HMZ16,HMZ14} to obtain our results, we first present it in this section.
Tree extraction, a key process of this framework, works as follows:
Given a tree $T$ and a subset, $X$, of its nodes including the root,
we construct a new tree by deleting the nodes of $T$ that are not in
$X$. Whenever a node $v$ is deleted, its children are inserted in its
place as the children of $v$'s parent, preserving the original
left-to-right order among nodes.
Figure~\ref{fig:extraction} shows an example.

\begin{figure}[t]
\includegraphics[width=\textwidth]{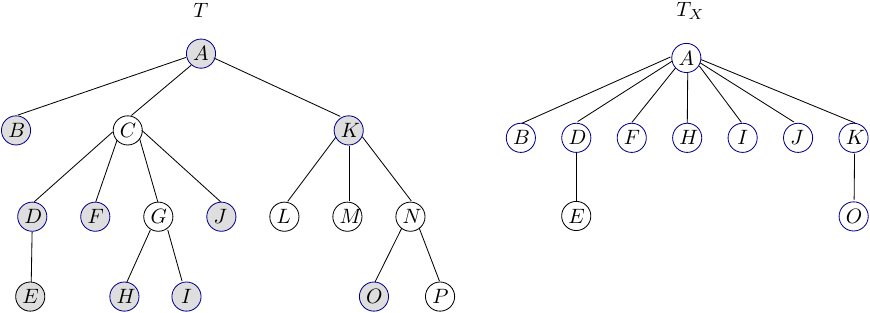}
\caption{An example of tree extraction, in which the original tree $T$ is shown on the
  left, $X=\{A, B, D, E, F, H, I, J, K, O\}$ is the set of shaded nodes, and the extracted tree $T_X$ is shown on the right.}
\label{fig:extraction}
\end{figure}

He, Munro and Zhou~\cite{HMZ14} use tree extraction to represent a labeled tree to support labeled navigational operations. Given an ordinal tree ${\cal T}$ with $|\T|$ nodes, each with a label from universe $\U = \{1,2,\ldots,u\}$, they define a subtree $\T_\alpha$ for each $\alpha \in \U$ by using on $\T$ the tree extraction process described in the previous paragraph, for the set $X_\alpha$ consisting of 
the nodes labeled $\alpha$, their parents and the root.
Thus, $\sum_{\alpha \in \U} |\T_\alpha| = O(|\T|)$. The nodes of the trees $\T_\alpha$ are marked with values $0$ or $1$, so that nodes marked $1$ correspond to nodes of $\T$ labeled $\alpha$. 
He et al.\ show how to compute, in $O(\log\log_\omega u)$ time, a mapping $f_\T(v,\alpha) = v'$, where $v$ is a node from $\T$ and $v'$ is a node from $T_\alpha$, so that the number of nodes labeled $\alpha$ from $v$ to the root of $\T$ equals the number of nodes labeled $1$ from $v'$ to the root of $\T_\alpha$.
Their representation is built in $O(|\T|\log u)$ time and
occupies $|\T|\lg u + O(|\T|) + o(|\T|)\cdot\lg u$ bits, or
$O(|\T|)$ words, of space. It supports these operations in time dominated by the cost to compute $f_\T$:
\begin{description}
    \item[$\mathtt{parent}_\alpha(v)$:] the lowest ancestor of $v$ labeled $\alpha$ (or $\perp$ if there is none).
    \item[$\mathtt{rank}_\alpha(v)$:] the number of ancestors of $v$ labeled $\alpha$.
    \item[$\mathtt{select}_\alpha(v,i):$] the $i$th highest ancestor of $v$ labeled $\alpha$.
\end{description}

In a later paper, He, Munro and Zhou~\cite{HMZ16} extend the concept of a {\em wavelet tree} \cite{GGV03,Nav14} to support more sophisticated operations called {\em path queries}. They partition the universe $\U = \U_{1,u}$ into subuniverses $\U_{a,b} = \{a,a+1,\ldots, b\}$ via successive halving: starting from $a=1$ and $b=u$, they partition $\U_{a,b}$ into $\U_{a,m}$ and $\U_{m+1,b}$, where $m= \lfloor (a+b)/2\rfloor$. They store trees $\T_{a,b}$ labeled with $0$ and $1$, where each node $v \in \T_{a,b}$ is labeled $0$ if the label of $v$ belongs to $\U_{a,m}$ and $1$ if it belongs to $\U_{m+1,b}$. Starting with $\T_{1,u} = \T$, they use tree extraction to define $\T_{a,m} = (\T_{a,b})_0$ and $\T_{m+1,b} = (\T_{a,b})_1$. The partition ends at the trees of the form $\T_\alpha = \T_{\alpha,\alpha}$, for $\alpha\in\U$. Because the universe is just \{0,1\}, they store each tree $\T_{a,b}$ using $O(|\T_{a,b}|)$ bits, so that the total hierarchical partition uses $O(|\T|\log u)$ bits, or $O(|\T|)$ space. It is built in time $O(|\T|\log u)$.
This arrangement can be used to answer the following queries:
\begin{description}
\item[$\mathit{Path~counting}$:] Given $v \in \T$ and labels $a \le b$, counts the number of ancestors of $v$ with label between $a$ and $b$.
    \item[$\mathit{Path~selection}$:] Given $v \in \T$ and an integer $i$, finds the ancestor of $v$ with the $i$th smallest label. 
\end{description}

This hierarchical partitioning of the universe enables $O(\log u)$ time support for path queries. He et al.\ then 
improve the query time to $O(\log_{\omega} u)$
by decomposing the universe into a sublogarithmic number of
subuniverses each time~\cite{HMZ16}.
\textcolor{black}{Their improved approach achieves succinct space simultaneously: when $u \le {\omega}^{\epsilon}$ for some
constant $\epsilon \in (0,1)$, the space cost is $|\T|(\lg u + 2) +o(|\T|)$
bits; otherwise, it uses $|\T|\lg u + O(\frac{|\T|\lg u}{\lg\lg |\T|})$ bits.}

\textcolor{black}{Note that this data structure can map any node $v$ of $\T$ to its lowest ancestor that appears in $\T_\alpha$ in $O(\log_{\omega} u)$ time for any label $\alpha$, thus supporting $\mathtt{parent}_\alpha(v)$ in $O(\log_{\omega} u)$ time. 
This mapping also reduces the support of $\mathtt{rank}_\alpha(v)$ and $\mathtt{select}_\alpha(v,i)$ to the computation of the depth and level ancestor of $\mathtt{parent}_\alpha(v)$ in $\T_\alpha$, which can be computed in constant time with $o(|\T_\alpha|)$ additional bits.
This means that the data structures of He et al.~\cite{HMZ16} can additionally support $\mathtt{parent}_\alpha(v)$, $\mathtt{rank}_\alpha(v)$, and $\mathtt{select}_\alpha(v,i)$ in $O(\log_{\omega} u)$ time.
 Although the query time of these three operators is more than that achieved by the data structures of He et al.~\cite{HMZ14}, it can be used to support both these labeled navigational operations and the path queries simultaneously without having to construct another set of data structures (e.g., \cite{HMZ16}), hence keeping the overall space succinct if desired.}

\section{Warm-up: Insertion Compressibility}
\label{sec:insertion}

\noindent
We start with a weaker compressibility definition that only captures the fact that some subsets can be (large) subsets of others, and thus can be used to represent the
containing subset within little space. This definition, as explained, is the dual of the 
``containment entropy'' \cite{ABGNP25}, which focuses on representing the contained set in terms of the containing one.

\begin{definition} \label{def:insentropy}
Let $\Sets$ be a set of sets. For each $S \in \Sets$, let $p(S)$ be a largest strict subset of $S$ in  $\Sets$, or $\emptyset$ if no such subset exists. Then the {\em insertion compressibility} of $\Sets$ is 
$$ \IEntropy(\Sets) ~~=~~ \sum_{S \in \Sets} |S \setminus p(S)| ~~=~~ \sum_{S \in \Sets} |S| - |p(S)|.
$$
\end{definition}

To represent a set of sets $\Sets$ within $O(\IEntropy(\Sets))$ space, we define
a so-called ``\textcolor{black}{setdiff tree}''.

\begin{definition}
Let $\Sets$ be a set of sets, with $p(S)$ according to Def.~\ref{def:insentropy}. The {\em \textcolor{black}{setdiff tree}} of $\Sets$ is a weighted directed graph with nodes $\Sets \cup \{\emptyset\}$ and an edge from each set $S \in \Sets$ to $p(S)$, with weight $|S| - |p(S)|$.
\end{definition}

It is easy to see that the \textcolor{black}{setdiff tree is} rooted at $\emptyset$ and its weights add up to $\IEntropy(\Sets)$; see the left of Figure~\ref{fig:ientropy}. By storing $S \setminus p(S)$ at each node $S$ of the \textcolor{black}{tree}, we spend $O(\IEntropy(\Sets))$ space and can reconstruct $\Sets$.

\begin{figure}[t]
\includegraphics[width=\textwidth]{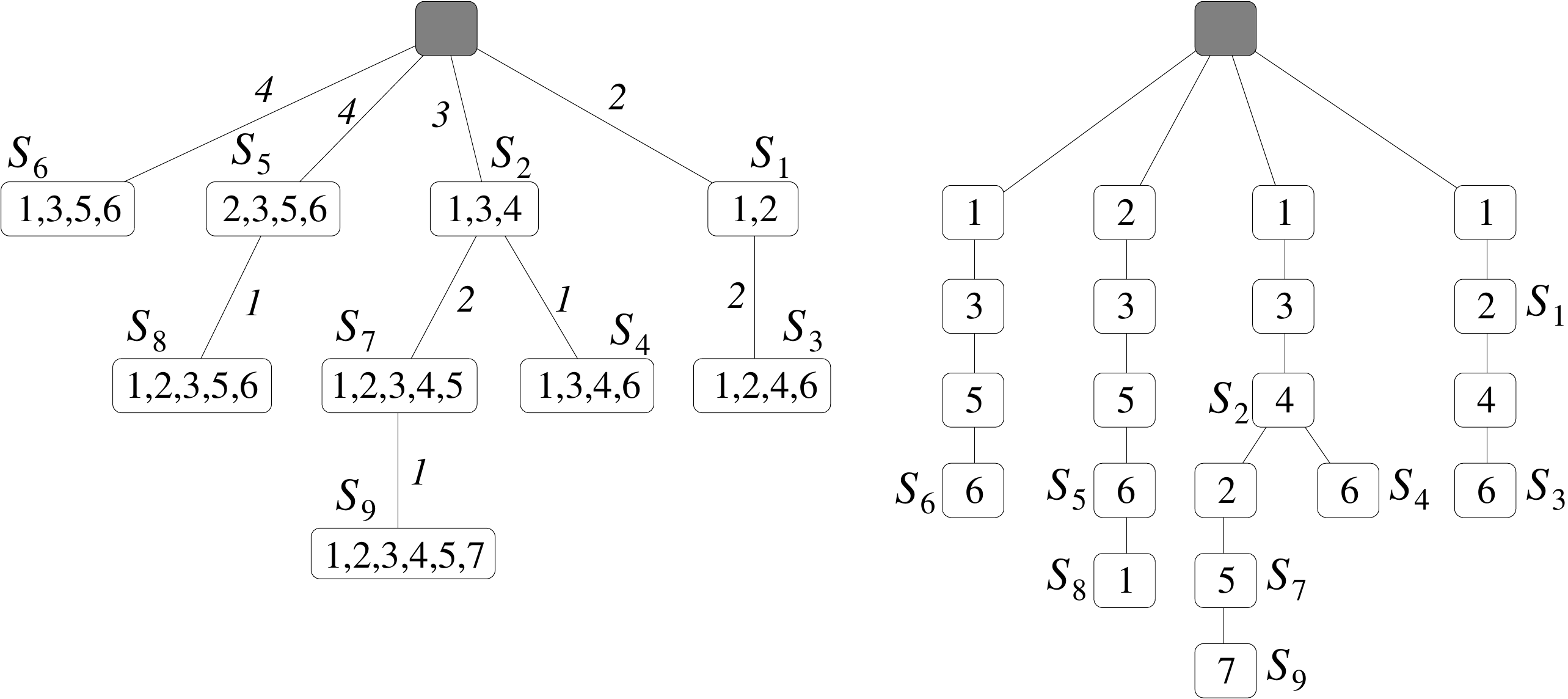}
\caption{On the left, a \textcolor{black}{setdiff tree} for sets $\Sets = \{ S_1,\ldots,S_9\}$, with the weights written as slanted values on the edges. It holds $\IEntropy(\Sets)=20$. On the right, a corresponding insertion tree with 21 nodes. We write $S_i$ besides its corresponding node $v(S_i)$.}
\label{fig:ientropy}
\end{figure}

We will use a slightly different representation, in the form of a labeled tree, still of size $O(\IEntropy(\Sets))$, in order to efficiently support the five queries on $\Sets$.

\begin{definition}
Let $\Sets$ be a set of sets, with $p(S)$ according to Def.~\ref{def:insentropy}. The {\em insertion tree} of $\Sets$ is
defined as follows.
\begin{itemize}
    \item The root of the tree is the empty set, and is called $v(\emptyset)$.
    \item Every set $S \in \Sets$ is a node, which we call $v(S)$, the node associated with $S$ (but there may be other tree nodes that are not associated with any sets).
    \item Let $I(S) = S \setminus p(S)$. Then the tree has a chain of $|I(S)|$ edges from $v(p(S))$ to $v(S)$, each labeled with a distinct element of $I(S)$, in arbitrary order.
\end{itemize} 
\end{definition}

The insertion tree of $\Sets$ clearly has $1+\IEntropy(\Sets)$ nodes; see the right of Figure~\ref{fig:ientropy}. We represent it in $O(\IEntropy(\Sets))$ space using the tree extraction framework.
More precisely, we represent the insertion tree using one labeled tree
structure \cite{HMZ14} to enable $O(\log\log_\omega|\U|)$-time
support of $\mathtt{parent}_\alpha$, $\mathtt{rank}_\alpha$ and
$\mathtt{select}_\alpha$, and we additionally represent it using another
labeled tree structure \cite{HMZ16} to support path counting and
path selection queries in $O(\log_{\omega} |\U|)$ time.
These two structures use $2|\IEntropy(\Sets)|\lg |\U| + O(|\IEntropy(\Sets)|) + o(|\IEntropy(\Sets)|)\cdot\lg |\U|$
bits in total, which is within $O(\IEntropy(\Sets))$ words of space. 
We now describe how the basic operations can be supported on any set $S$.

\begin{description}
\item[Membership] To find out whether $x \in S$, we locate $v(S)$ (which we can directly associate with $S$ within the space budget). We then reduce the query to the primitive $\mathtt{parent}_x(v(S))$ \cite{HMZ14} (which finds the closest ancestor of $v(S)$ labeled $x$), returning true iff there is one. This takes time $O(\log\log_\omega|\U|)$.
\item[Access] To access the $i$th smallest element of $S$, we retrieve the $i$th smallest symbol on an edge in the 
path from $v(S)$ to the root. This corresponds to a path selection query in our insertion tree, requiring time $O(\log_{\omega} |\U|)$.
\item[Rank] A rank for $x$ in $S$ corresponds to counting the number of labels of ancestors of $v(S)$ that are at most $x$. This is solved with a path counting query with interval $[1,x]$ in $O(\log_{\omega} |\U|)$ time.
\item[Predecessor/Successor] To find the predecessor of $x$ in $S$ we compute the rank $i$ of $x$ in $S$. If $i=0$, then $x$ has no predecessor in $S$; otherwise we obtain it by accessing the $i$th smallest element of $S$. For successor, if the predecessor of $x$ is not $x$, we access and return instead the $(i+1)$st smallest
element of $S$, or return that there is no successor if $i$ is the depth of $v(S)$ in the insertion tree.
This entire process uses $O(\log_{\omega} |\U|)$ time.
\end{description}

The insertion tree is easily built from the \textcolor{black}{setdiff tree} in time $O(\IEntropy(\Sets))$; the costly part is to find $p(S)$ for every $S \in \Sets$. Let $s = |\Sets|$, $n = \sum_{S \in \Sets} |S|$, and $u=|\U|$. We can first sort the sets by increasing size in time $O(s \log s)$, and insert them one by one as nodes in the \textcolor{black}{setdiff tree} (which initially contains only the node $\emptyset$) with an edge towards $p(S)$. To insert $S$, we sort its elements (which induces total time $O(n\log u)$) and then traverse the current \textcolor{black}{setdiff tree} in DFS order from the node $\emptyset$, finding the largest sets $S' \subset S$. Checking such containments takes time $O(|S|+|S'|)$ with a merge-like algorithm. This process takes time $O(\sum_{S,S'} (|S|+|S'|)) = O(sn)$. If we store the sets $S\setminus p(S)$ along this same process, the \textcolor{black}{setdiff tree} is built in additional time $O(\IEntropy(\Sets)) \subseteq O(n)$. From the \textcolor{black}{setdiff tree}, we build the insertion tree in time $O(\IEntropy(\Sets))$ and the tree extraction data structure in additional time $O(n\log u)$.

Tree extraction assumes that $\U=[1.. u]$. If this is not the case, we collect all the $n$ elements, sort them in $O(n\log u)$ time, and assign \textcolor{black}{them} to the $u$ distinct integers in $[1..u]$. Queries and answers can be translated in constant time.

\begin{theorem} \label{thm:insertion}
On a word RAM with $\omega$-bit words, a set $\Sets$ of $s$ sets, over a universe of size $u=|\cup_{S \in \Sets} S|$, can be represented within $O(\IEntropy(\Sets))$ space so that access, rank, predecessor and successor queries can be carried out in time $O(\log_\omega u)$ and membership in time $O(\log\log_\omega u)$. 
If $n=\sum_{S \in \Sets} |S|$, then the data structure can be built in time $O(n\log u + sn)$.
\end{theorem}

\textcolor{black}{When $\U$ is a set of integers from a bounded universe, we can achieve succinct space by constructing only the succinct path query data structure of He et al.~\cite{HMZ16} to represent the insertion tree of $\Sets$ without constructing their labeled tree structure~\cite{HMZ14}, and uniquely identifying each set $S$ in $\Sets$ by the preorder number of the node representing it in the insertion tree; this is how nodes are referred to in the succinct path query data structure.
As commented in Section~\ref{sec:tree-extraction}, this single data structure can still be used to support $\mathtt{parent}_\alpha(v)$, $\mathtt{rank}_\alpha(v)$, and $\mathtt{select}_\alpha(v,i)$ in $O(\log_{\omega} u)$ time.
By our description of query algorithms in this section, this would slow down the support of membership to $O(\log_\omega u)$, while the support of all other operators remains the same. Thus, we have the following corollary.}

\begin{corollary} 
 \label{cor:succinctinsertion}
\textcolor{black}{On a word RAM with $\omega$-bit words, a set $\Sets$ of $s$ sets over the universe $\{1, 2, \ldots, u\}$ can be represented within $|\IEntropy(\Sets)|\lg u + O(|\IEntropy(\Sets)|) + o(|\IEntropy(\Sets)|)\cdot\lg u$ bits of space so that membership, access, rank, predecessor and successor queries can be carried out in time $O(\log_\omega u)$. 
If $n=\sum_{S \in \Sets} |S|$, then the data structure can be built in time $O(n\log u + sn)$.}
\end{corollary}

In the next section we give, as a byproduct, an improved construction time; see Theorem~\ref{thm:Iconstr}.

\section{Symmetric-Difference Compressibility}

\noindent
We now give our definition of symmetric-difference (or {\em symdiff} for short) compressibility, which allows expressing a set by both inserting in or removing elements from some other set. This is strictly stronger than both insertion and deletion compressibility. \textcolor{black}{From here on, $\U$ will denote the effective universe of the elements $S \in \Sets$.}

\begin{definition} \label{def:setdifgraph}
A {\em \textcolor{black}{symdiff forest}} on a set of sets $\Sets$ is a weighted directed graph on the nodes $\Sets \cup \{ \emptyset, \U\}$, where $\U = \cup_{S \in \Sets} S$. There is exactly one edge leaving from each set $S \in \Sets$. The target node of that edge is called $p(S)$ and the weight of the edge is $|S \triangle p(S)|$, the size of the symmetric difference. 
Further, there is a path from every $S \in \Sets$ to $\emptyset$ or to $\U$. The {\em weight} of a \textcolor{black}{symdiff forest} is the sum of the weights of its edges.
\end{definition}

Note that a \textcolor{black}{symdiff forest is always formed by two trees}, one rooted at $\emptyset$ and the other rooted at $\U$; see the left of Figure~\ref{fig:dentropy} \textcolor{black}{(the zero-weight edge between $\emptyset$ and $\U$ will be added later)}.
Given a \textcolor{black}{symdiff forest} for $\Sets$, we can represent at each node $S$ only those elements in $S \setminus p(S)$ and those in $p(S) \setminus S$, so that we can reconstruct $S$ from $p(S)$. The total space incurred is the weight of the two trees. This is a slight variation over the model of Alves et al. \cite{AMB+25}, which use only one tree rooted at $\emptyset$.

\begin{figure}[t]
\includegraphics[width=\textwidth]{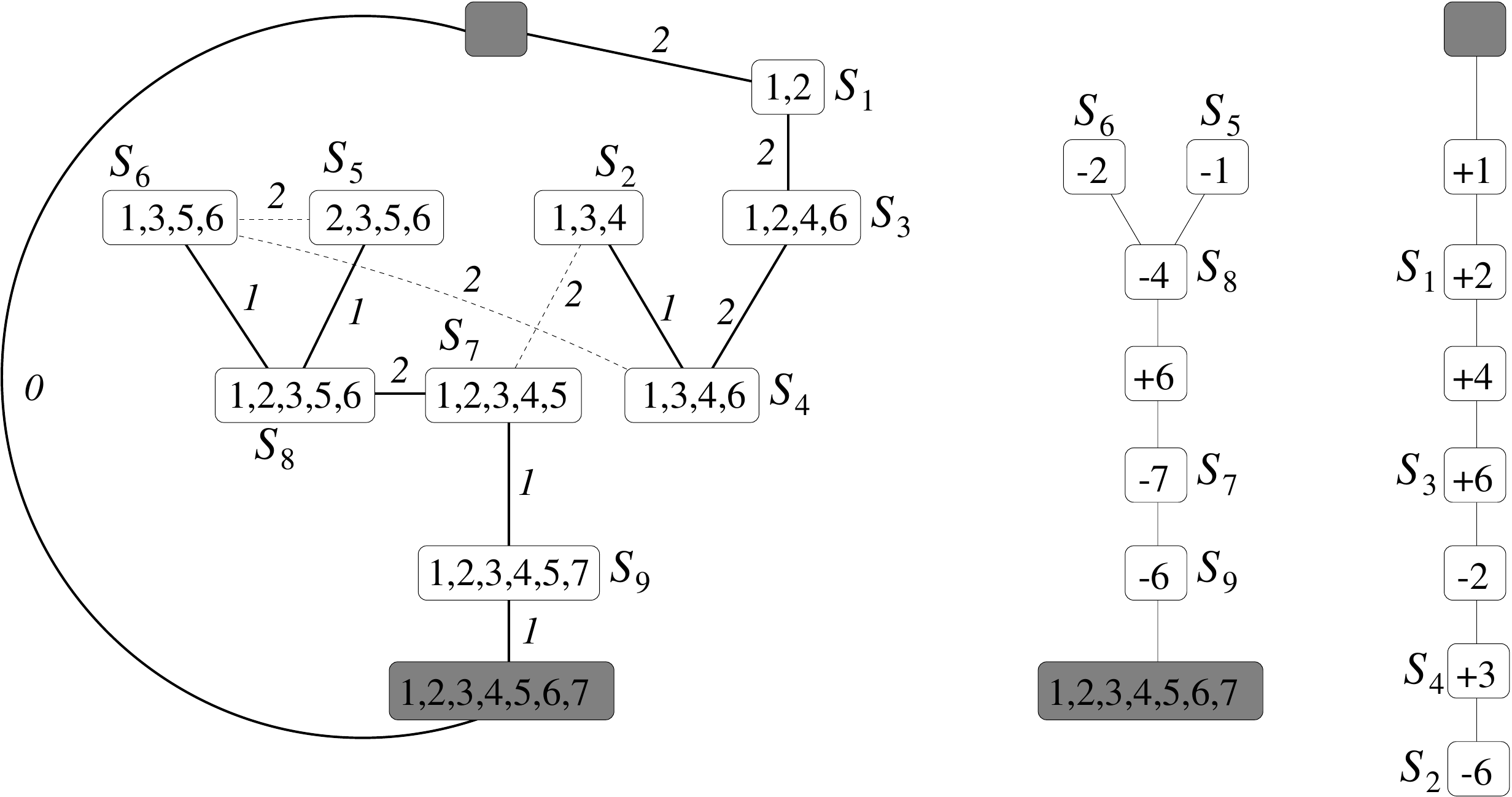}
\caption{On the left, a \textcolor{black}{symdiff forest} of minimum weight between the same sets of Figure~\ref{fig:ientropy}, yielding $\DEntropy(\Sets) = 13$. \textcolor{black}{An MST (also of weight $13$) has the same edges of the forest plus the zero-weight edge between $\emptyset$ and $\U$.} To build it, it suffices to consider the full graph edges with weights up to $\ell=2$. Among these edges, those not in the MST are dashed. On the right, the indel trees, rooted at $\U$ (upside down) and at $\emptyset$, that represent $\Sets$ in space $O(\DEntropy(\Sets))$. We write $+x$ and $-x$ to indicate elements $x$ to add or to delete, respectively, from the parent node.}
\label{fig:dentropy}
\end{figure}

\begin{definition}
The {\em symdiff compressibility} of a set of sets $\Sets$, $\DEntropy(\Sets)$, is the minimum weight of a \textcolor{black}{symdiff forest} on $\Sets$.    
\end{definition}

We now describe a representation that uses $O(\DEntropy(\Sets))$ space and supports the five basic queries in the almost same times as with insertion compressibility.

As explained, we will represent the two trees, storing at each node the symmetric difference with its parent node. We will later modify this tree and add more structures in order to support efficient queries on the sets. 

\subsection{Construction}

\noindent
Analogously to what Alves et al.~\cite{AMB+25} do on a slightly different model,
we can build the lightest \textcolor{black}{symdiff forest} as the minimum spanning tree (MST) over the full graph with nodes $S \in \Sets \cup \{\emptyset,\U\}$ and edge weights $w(S,S')=|S \triangle S'|$. The only exception is that, in order to obtain the best pair of trees according to our model, we set $w(\emptyset,\U)=0$. We first show the correctness of this construction.

\begin{lemma}
The weight of the described MST is $\DEntropy(\Sets)$ and the \textcolor{black}{symdiff forest} that achieves it are obtained by removing the edge between $\emptyset$ and $\U$.
\end{lemma}
\begin{proof}
Every MST contains the lightest edge, which is the one between $\emptyset$ and $\U$ \textcolor{black}{because all the others have positive weights}. Removing this edge disconnects the MST into two trees, one rooted at $\emptyset$ and one rooted at $\U$. The sum of the weights of the two trees is the same of the MST, as we removed a zero-cost edge. On the other hand, any pair of trees rooted at $\emptyset$ and $\U$ can be turned into a spanning tree of the same cost by adding the zero-cost edge between $\emptyset$ and $\U$. It follows that the cost of the MST is $\DEntropy(\Sets)$.
\end{proof}

We can build the MST in time $O(s^2)$ once all the edge weights are computed as in Section~\ref{sec:insertion}: increasingly sort the elements of all sets in time $O(n\log u)$, and compute all the edge weights in time $O(sn)$. The total construction time is then $O(n\log u + sn)$ as for insertion compressibility. This is only slightly better than the complexity obtained by Alves et al.~\cite{AMB+25}: $O(sn+s^2\log s)$ if the sets come already sorted. 
We now show that a more refined construction time is possible. 

We first show that, once the sets are sorted, one can compute $S \triangle S'$ in time $O(|S \triangle S'|)$ after an $O(n \log u)$-time preprocessing of $\Sets$. We map the elements of $\U$ to the interval $[1 .. u]$ in time $O(n\log u)$ if necessary, as before. Let $S_1,\ldots,S_s$ be the mapped sets. We then build a {\em text} $T(\Sets) = S_1\$_1S_2\$_2\cdots S_s\$_s$, which is a string where all the mapped sets $S_i$ are appended a unique terminator symbol $\$_i = u+i$ and concatenated. Since $|T(\Sets)| = n+s = O(n)$, we can build the suffix tree \cite{Apo85,CR02} of $T(\Sets)$ in
$O(n)$ time \cite{FFM00}, with support for constant-time level ancestor queries \cite{BF04}: $lca(u,v)$ is the deepest ancestor of suffix tree leaves $u$ and $v$, and its (stored) string depth is the length of the longest common prefix of the suffixes denoted by those leaves. By storing the suffix tree leaf corresponding to each suffix $T(\Sets)[i..]$, and the position where each string $S_i \in \Sets$ starts in $T(\Sets)$ (all of which takes $O(n)$ space) we can, in constant time, compute $lcp(S[p..],S'[q..])$, that is, how many equal symbols follow from any $S[p..]$ and any $S'[q..]$.

With this tool we can compare $S$ and $S'$ as follows. We start from $p,q:=1$ and find the longest common prefix $lcp(S[p..],S'[q..])=\ell$ in constant time. This means that the first position where the strings differ is $S[p+\ell] \neq S'[q+\ell]$. The smaller of both symbols is the first element of $S \triangle S'$. If the smaller element is $S[p+\ell]$ we set $p:=p+\ell+1$ and $q:=q+\ell$; if it is $S'[q+\ell]$ we set $p:=p+\ell$ and $q:=q+\ell+1$. We then continue extracting the elements of $S \triangle S'$ one by one. When both $p$ and $q$ reach terminators $\$_*$, we are done. Note that we can use this technique in iterator mode, finding the next difference at each call, so that we determine in time $O(k)$ whether the symmetric difference exceeds $k$.

\begin{lemma} \label{lem:iter}
After $k$ calls to this procedure on $S$ and $S'$, if both $p$ and $q$ point to terminators $\$_*$, then $|S \Delta S'| < k$. Otherwise, there are $k$ elements of $S \Delta S'$ listed in $S[..p-1]$ and $S'[..q-1]$, so $|S \Delta S'| \ge k$. 
\end{lemma}
\begin{proof}
This holds after $k=0$ calls because $p=q=1$ and it cannot be $S=S'=\emptyset$. Now assume the proposition holds after $k-1$ calls. \textcolor{black}{The $k$th call occurs only if $p$ and $q$ do not both point to terminators} $\$_*$. The $k$th call then finds $\ell = lcp(S[p..],S'[q..])$, meaning there is no element of $S \Delta S'$ in $S[p..p+\ell-1]$ or $S[q..q+\ell-1]$, but $S[p+\ell] \neq S'[q+\ell]$. Because the elements of $S$ and $S'$ are increasingly sorted, this means that the smallest of $S[p+\ell]$ and $S'[q+\ell]$ is in $S \Delta S'$, or both are terminators $\$_*$. Assume $S[p+\ell] \in S \Delta S'$; the case $S'[q+\ell] \in S \Delta S'$ is analogous. The algorithm sets $p := p+\ell+1$ and $q := q+\ell$, which restablishes the property because now there are $(k-1)+1$ elements of $S \Delta S'$ in $S[..p-1]$ and $s'[..q-1]$. If, on the other hand, both $S[p+\ell]$ and $S'[q+\ell]$ are terminators $\$_*$, this means there are no more elements in $S \Delta S'$ and therefore $|S \Delta S'|=k-1<k$.
\end{proof}

We adapt Prim's 
algorithm, which grows a set $V$ (the nodes already attached to the MST), by taking at each step a node out of $V'$ (the nodes not yet in the MST) and moving it into $V$. Initially we set $V := \{ \emptyset, \U\}$, with the zero-cost edge connecting them in the MST, and $V' = \Sets$. Unlike the classic algorithm, which knows all the weights upfront, we will start with all edge weights set to $+\infty$ for Prim's algorithm (except $w(\emptyset,\U)=0$, which is already in the MST), and will discover the true weights incrementally, from smallest to largest. 

\textcolor{black}{We run Prim on the full graph.} We will maintain all the edges whose weights are yet unknown in a {\em bag} (initially of size ${s+2 \choose 2} - 1$). For each edge $(S,S')$ in the bag, we initialize its iterator at $p,q := 1$, which completes the iteration $k=0$.

Now we start the iterations $k=1$ onwards.
In the $k$th iteration, we advance the iterators once in all the edges $(S,S')$ of our bag. For those edges where both $p$ and $q$ end up pointing to terminators $\$_*$, we know that their weight is $w(S,S')=k-1$, as by Lemma~\ref{lem:iter} it was $\ge k-1$ and it is $<k$. So we set the corresponding weight for Prim's algorithm and remove the edge from the bag.

The invariant is that after the $k$th iteration we have defined all the edge weights that are less than $k$ (note that those edges may or may not connect nodes from $V$ and $V'$, so not all of them are immediately useful for Prim). We can then run some steps of Prim, until the next lowest weight to include in the MST is $+\infty$. At that point we suspend Prim's algorithm and move on to the next iteration, the $(k+1)$th, where we find all the weights equal to $k$. The next lemma proves the correctness of this approach.

\begin{lemma}
Prim's algorithm runs correctly by using only the edges of weight up to $k$ before using any heavier edge.
\end{lemma}
\begin{proof}
Let $E = S_k \cup L_k$ be the graph edges, with $S_k = \{ e \in E,~w(e) < k\}$ being the edges with weight less than $k$. At any point in Prim's algorithm, it chooses the edge $e_{\min} = \mathrm{arg min} \{ w(e),~ e \in E \cap (V \times V') \}$, that is, the lightest one connecting $V$ and $V'$, and includes $e_{\min}$ in the MST (and one of its extremes in $V$). This can be rewritten as $e_{\min} = \mathrm{arg min} \{ w(e),~ e \in (S_k \cup L_k) \cap (V \times V') \}$. Since $w(e) < k \le w(e')$ for all $e \in S_k$ and $e' \in L_k$, it follows that $e_{\min} = \mathrm{arg min} \{ w(e),~ e \in S_k \cap (V \times V') \}$ if $S_k \cap (V \times V') \neq \emptyset$. It is therefore correct to run Prim only on the edges $S_k$ discovered up to iteration $k$ as long as it finds candidate edges of weight less than $k$.
\end{proof}

For the analysis, note that, if $\ell$ is the heaviest weight included in the MST, then we have incremented each weight $w(S,S')$, along the algorithm,  $\min(\ell,|S \triangle S'|)$ times. The total time is then $O(\sum_{S,S'} \min(\ell,|S \triangle S'|))$, which is bounded both by $O(s^2\ell)$ and by $O(sn)$ (the latter because $|S \triangle S'| \le |S|+|S'|$).

\begin{theorem} \label{thm:constr}
A set of sets $\Sets$ can be represented within $O(\DEntropy(\Sets))$ space. If $\Sets$ has $s$ sets, the sum of the sizes of its sets is $n$, they contain $u$ distinct elements in total, and the maximum weight in a minimum-weight \textcolor{black}{symdiff forest} of $\Sets$ is $\ell$, then that \textcolor{black}{forest} can be built in time $O(n\log u + \sum_{S,S' \in \Sets} \min (\ell,|S \triangle S'|)) \subseteq O(n\log u + \min(s^2 \ell,sn))$.
\end{theorem}

\textcolor{black}{Alternatively, one can use a recent Monte Carlo construction~\cite[Cor.~4]{KLP26}, which builds the MST of $s$ bitvectors of dimension $u$, using Hamming distance for the edge weights, in worst-case time $O(s(s+u+W)\log u)$ and correctly with high probability, where $W$ is the weight of the resulting MST. By converting our sets into  bitvectors in time $O(su)$ (i.e., the bitvector for $S$ has 1s at the positions $x \in \Sets$) and adding bitvectors $0^u$ and $1^u$, their algorithm finds the symdiff forest, with $W=\Delta(\Sets)$, in total time $O(s(u+\Delta(\Sets))\log u)$ (note $s \le \Delta(\Sets)+1$). Our construction is deterministic and takes time $O(n\log u + s^2\ell)$, where $s\ell \ge \Delta(\Sets)$ (but $s\ell=\Theta(\Delta(\Sets))$ may hold), so the complexities are incomparable.}

\paragraph{Improving the computation of $\IEntropy(\Sets)$}

A similar approach can speed up the construction of the \textcolor{black}{setdiff tree} of Section~\ref{sec:insertion}. For every new set $S$, instead of fully computing the distance towards all the preceding sets $S'$, we create a bag with those and use the same iterations from $k=0$ until finding the
first value $k-1$ associated with some set $S'$ such that $S' \subset S$ and $|S \setminus S'| = k-1$. We then set $p(S) := S'$. The computation using the pointers $p$ and $q$ is a bit different, because whenever we find the distinct element to occur in $S'$ we must discard $S'$ because $S' \not\subset S$. Overall, we spend time $O(s \cdot |S \setminus p(S)|)$ to add $S$ to the tree.

\begin{theorem} \label{thm:Iconstr}
The data structure of Theorem~\ref{thm:insertion} can be built in time $O(n\log u + s \cdot \IEntropy(\Sets)) \subseteq O(n\log u + sn)$.
\end{theorem}

By symmetry, this construction can be used to compute the containment entropy \cite{ABGNP25}, which is based on deletions instead of insertions, and has $\U$ instead of $\emptyset$ at the root. Their original construction takes time $O(sn\log n)$ \cite{ABGNP25}.

\subsection{Supporting the operations}{\label{sec:ops}}

\noindent
We will use the following data structure to represent $\Sets$ in space $O(\DEntropy(\Sets))$. See the right part of Figure~\ref{fig:dentropy}.

\begin{definition}
Let $\Sets$ be a set of sets, with $p(S)$ according to Def.~\ref{def:setdifgraph}. The {\em indel trees} of $\Sets$ are two trees defined as follows.
\begin{itemize}
    \item The root of one tree, called $v(\emptyset)$, represents the empty set, and the root of the other, called $v(\U)$, represents $\U$.
    \item Every set $S \in \Sets$ is a node, which we call $v(S)$, in the same tree of $v(p(S))$ (but there may be other tree nodes that are not associated with any sets).
    \item Let $I(S) = S \setminus p(S)$ and $D(S) = p(S) \setminus S$. Then the tree where $v(S)$ belongs has a chain of $|I(S)|+|D(S)|$ edges from $v(p(S))$ to $v(S)$, each labeled with a distinct element of $I(S)$ or $D(S)$, in arbitrary order. Those labels $x \in I(S)$ are written $+x$ and those $x \in D(S)$ are written $-x$.
\end{itemize} 
\end{definition}

It is clear that the two indel trees of a minimum-weight \textcolor{black}{symdiff forest} have $\DEntropy(\Sets)+2$ nodes in total. 
Again we represent either tree using the labeled tree data structures
\cite{HMZ14,HMZ16} to support labeled operations and path queries, and
this time the labels are from the set $[-u..u]$.
These two structures then use $2|\DEntropy(\Sets)|\lg u + O(|\DEntropy(\Sets)|) + o(|\DEntropy(\Sets)|)\cdot\lg u$
bits in total, which is within $O(\DEntropy(\Sets))$ words of space. 
We now show how the set queries are carried out using operations on labeled trees.

\begin{description}
\item[Membership] To find whether $x \in S$, we first locate $v(S)$. We then compute $u^+ = \mathtt{parent}_{+x}(v(S))$ and $u^- = \mathtt{parent}_{-x}(v(S))$ \cite{HMZ14}. If both exist, then $x \in S$ iff the depth of $u^+$ is larger than that of $u^-$ (i.e., the last edit operation involving $x$ was an insertion). If only $u^+$ exists, then $x \in S$. If only $u^-$ exists, then $x \not\in S$. If none \textcolor{black}{exist}, then $x \in S$ iff $v(S)$ descends from $v(\U)$. As for insertion compressibility, this process takes $O(\log\log_\omega|\U|)$ time.
\item[Access] We show in Section~\ref{sec:newop} how to solve this
  operation in time $O(\log|\U|)$ with $O(\DEntropy(\Sets))$ words of
  additional space.
\item[Rank] The rank of $x$ in $S$ is found by counting the number of labels of ancestors of $v(S)$ in the range $[1..x]$, minus the number of labels of ancestors of $v(S)$ in the range $[-x..-1]$. If $v(S)$ descends from $v(\U)$, we add $x$ to this difference, because the elements $[1..u]$ are tacitly assumed to exist at the root. This computation is done with two path counting queries, in $O(\log_{\omega} |\U|)$ time.
\item[Predecessor/Successor] This is solved exactly as for insertion compressibility. As they use the access operation, their time is $O(\log|\U|)$.
\end{description}

\begin{theorem} \label{thm:indels}
On a word RAM with $\omega$-bit words, a set $\Sets$ of $s$ sets, over a universe of size $u=|\cup_{S \in \Sets} S|$, can be represented in $O(\DEntropy(\Sets))$ space so that membership queries can be performed in time $O(\log\log_\omega u)$, rank can be performed in time $O(\log_\omega u)$, and access, predecessor and successor in time $O(\log u)$. If $n=\sum_{S \in \Sets} |S|$, then, once a minimum-weight \textcolor{black}{symdiff forest} is constructed (see Theorem~\ref{thm:constr}), this structure can be built in $O(n\log u)$ extra time.
\end{theorem}

\textcolor{black}{We can use the same strategy for Corollary~\ref{cor:succinctinsertion} to represent $\Sets$ using the succinct path query structure only. Furthermore, if we do not construct the additional data structures in Section~\ref{sec:newop}, we can support access by performing a binary search over all elements, and each time we use rank to verify whether we have found the right element. Thus we have the following succinct data structure.}

\begin{corollary}\label{cor:indels}
\textcolor{black}{On a word RAM with $\omega$-bit words, a set $\Sets$ of $s$ sets over the universe $\{1, 2, \ldots, u\}$ can be represented in 
$|\DEntropy(\Sets)|\lg u + O(|\DEntropy(\Sets)|) + o(|\DEntropy(\Sets)|)\cdot\lg u$ bits of space so that membership and rank can be performed in time $O(\log_\omega u)$, and access, predecessor and successor in time $O(\log u\log_\omega u)$. If $n=\sum_{S \in \Sets} |S|$, then, once a minimum-weight \textcolor{black}{symdiff forest} is constructed (see Theorem~\ref{thm:constr}), this structure can be built in $O(n\log u)$ extra time.}
\end{corollary}

\textcolor{black}{It is worth mentioning that our support of access and rank operations over indel trees is a generalization to the problem of supporting segment rank and segment select operations over a set of $n$ horizontal line segments whose endpoints have coordinates in rank space~\cite{bgpt25}. 
Given two integers $y$ and $i$, a {\em segment rank} operation computes the number of segments with $y$-coordinates at most $y$ that cross a vertical line $x = i$, while a {\em segment select} operation returns the segment with the $j$-th smallest $y$-coordinate among all segments that cross a vertical line $x = i$, for any given $i$ and $j$.
Bille et al.~\cite{bgpt25} designed a data structure occupying $2n \log n + o(n \log n)$ bits of space that can support both operations in $O(\log n / \log\log n)$ time.}

\textcolor{black}{To see the connection with our data structure, consider the special case in which indel tree $\T$ rooted at $v(\emptyset)$ is a path and the other indel tree is empty. 
Consider a node $u$ with a positive label $+x$, and let $v$ be the lowest node such that between and including nodes $u$ and $v$, there is no node labeled $-x$. 
Then, if $u$ and $v$ are at levels $i$ and $j$ of $\T$, respectively, we create a line segment whose endpoints are $(i, u)$ and $(j, u)$.
In this way, we can transform $\T$ (when it is a path) to a set of line segments, and the rank and access operations over the sets represented by tree nodes respectively become the segment rank and segment select operations over this set of line segments.
Thus, our problem generalizes the segment rank and segment select problem to tree structured data.}

\subsection{Accessing with positive and negative values} \label{sec:newop}

\begin{figure}[t]
\centerline{\includegraphics[width=0.8\textwidth]{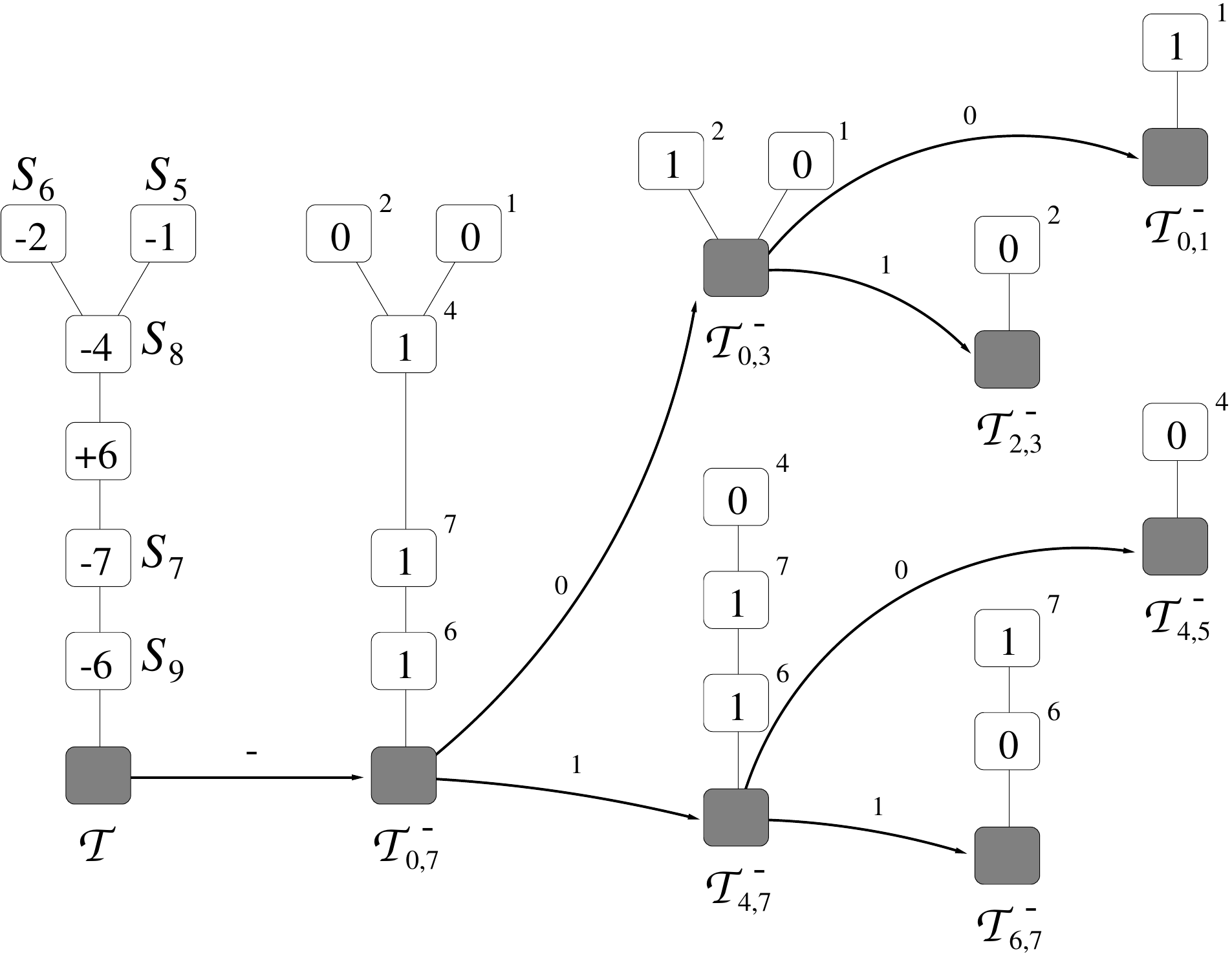}}
\caption{The hierarchical extraction process for $\T^- = \T^-_{0,7}$, starting from the (upside down) tree $\T$ rooted at $\U$ of Figure~\ref{fig:dentropy}. The rightward bold arrows lead from $\T^-_{a,b}$ to $\T^-_{a,m}$ (with label $0$) and to $\T^-_{m+1,b}$ (with label 1). We show the labels of the trees $\T^-_{a,b}$ inside the nodes, and the original symbols in small font near the boxes.}
\label{fig:wtree}
\end{figure}

\noindent
We now show how to access the $i$th smallest element of a set $S$. 
Among the two indel trees representing $\Sets$, let $\T$ be the one containing the node $v=v(S)$. 
Let $x$ be the element of $S$ whose rank is $i$, that is, the element we are looking for.
Then, because negative values $-x$ can only occur if $x$ is already present in the set that an ancestor of $v$ represents, the rank $i$ of $x$ is the number of ancestors of $v$ with labels in $[1..x]$ minus the number of ancestors of $v$ within labels in $[-x..-1]$) if $\T$ is rooted at $v(\emptyset)$; otherwise, it is this difference plus $x$. 
Thus, it is easy to use rank operations from $v$ to binary search for $x$ in the indel tree in time $O(\log u\log_\omega u)$. We will, instead, exploit the structure of the hierarchy described in Section~\ref{sec:tree-extraction} to do the binary search in time $O(\log u)$.

To achieve this, more preprocessing is required. For either indel tree
$\T$, we will use the hierarchical binary universe partitioning
described in Section~\ref{sec:tree-extraction} to form two
hierarchies, one for positive values $[1..u]$ starting at a tree
$\T^+$ and another for negative values $[-u..{-1}]$ (seen as $[1..u]$)
starting at a tree $\T^-$. Both trees $\T^+$ and $\T^-$ are
constructed from $\T$ via tree extraction (with labels $\{+,-\}$),
using function $f_\T$. Figure~\ref{fig:wtree} illustrates the process.
Since, for either indel tree, we perform $O(\lg u)$ levels of
partitioning, and, for each level, we construct a 0/1-labeled tree
which occupies $O(|\DEntropy(\Sets)|)$ bits, the extra space cost
incurred here is $O(|\DEntropy(\Sets)|\lg u)$ bits, or
$O(|\DEntropy(\Sets)|)$ words.

The main idea of our query algorithm is to start with ranges $[a..b] := [1..u]$ and, as we descend down both hierarchical partitions, shrink the difference between $a$ and $b$, while maintaining the invariant that $i$ is between the rank of $a$ in $S$ and the rank of $b$ in $S$.
To describe this algorithm in detail, we first consider the case in which $\T$ is rooted at $v(\emptyset)$. Letting $v=v(S) \in \T$, we start from $v^+ = f_{\T}(v,+) \in \T_{1,u}^+$ and $v^- = f_{\T}(v,-) \in \T_{1,u}^-$.
In this process, $v^+ \in \T_{a,b}^+$ and $v^- \in \T_{a,b}^-$ will be the current nodes in both hierarchical partitions. We compute $s^+ = \mathtt{rank}_0(v^+)$ in $\T_{a,b}^+$ and $s^- = \mathtt{rank}_0(v^-)$ in $\T_{a,b}^-$. Note that $\mathtt{rank}_0$ is computed in constant time because the alphabet, $\{0,1\}$, of $\T_{a,b}^\pm$ is of size two. 
Observe that, in the path of $\T$ between $v(\emptyset)$ and $v$, the number of nodes representing the insertion of an element must be either equal to or exactly one more than the number of nodes representing its deletion. It then follows that $S$ has $s^+ - s^-$ elements in $[a..m]$, where $m=\lfloor(a+b)/2\rfloor$. Therefore, if $s^+ - s^- \ge i$, then $x \in [a..m]$. In this case, we set $b := m$, $v^+ := f_{a,b}^+(v^+,0) \in \T_{a,m}^+$, and $v^- := f_{a,b}^-(v^-,0) \in \T_{a,m}^-$, where $f_{x,y}^{\pm}$, the mapping function of the tree $\T_{x,y}^{\pm}$, is again computed in constant time because the trees have only two labels. Otherwise, we set $a := m+1$, $v^+ := f_{a,b}^+(v^+,1)$, $v^- := f_{a,b}^-(v^-,1)$, and $i := i - (s^+ - s^-)$. We iterate until we reach $a=b$, when we return the answer $x:=a=b$.

When the tree root is $\U$ instead of $\emptyset$, we modify the above procedure by using $s^+-s^- + (m-a+1)$ instead of $s^+-s^-$.
As we spend $O(1)$ time on each level of $\T^+$ and $\T^-$, which have $O(\log u)$ levels, this process uses $O(\log u)$ time.

\section{\textcolor{black}{Computing the Support}} \label{sec:global}

\noindent
The operations we have studied so far concern a single set $S \in \Sets$. We study in this section a different kind of operation: the set $\supp(x,\Sets)$ defined in the Introduction, that is, the support of $x$. We show how to compute the support in optimal time, within space $O(\Delta(\Sets))$. We will also show that the cardinality of the support can be returned in constant time within this space budget.

\textcolor{black}{The key idea to compute $\supp(x,\Sets)$ is to maintain a tree $\T'$ formed with only the nodes $v(S) \in \T$, that is, those that represent some set $S \in \Sets$, respecting the ancestorship relations (this is isomorphic to the tree rooted at $\emptyset$ or at $\U$ we build from the symdiff forest). At query time, if the nearest ancestor of a node $v(S) \in \T$ labeled $+x$ or $-x$, is labeled $+x$, we will {\em mark} the corresponding node $v'(S) \in \T'$ with a ``$+$''; if the nearer label is $-x$, we will mark $v'(S)$ with a ``$-$''. We will then start a depth-first traversal from every node marked ``$+$'' in $\T'$, stopping when we find a leaf or a node marked ``$-$'' (for $\T_\U$, we also assume the root is marked ``$+$''). Achieving optimal time requires some refinements over this basic strategy, however.}

\subsection{Data structures}

\noindent
We will store $O(\Delta(\Sets))$-sized additional data on top of each of the two indel trees of $\Sets$: $\T = \T_\emptyset$ (rooted at $\emptyset$) and $\T = \T_\U$ (rooted at $\U$).
\begin{enumerate}
\item \textcolor{black}{We create the described tree $\T'$.}
Some arbitrary ordering is chosen among the children of the nodes in $\T'$.
\item Every node $u \in \T$ stores a pointer $down(u) = v'(S) \in \T'$, where $v(S)$ is the closest descendant of $u$ that is $v(S)$ for some $S \in \Sets$. This is well defined because there is a single path from $u$ to $v(S)$. In particular, $down(v(S)) = v'(S)$ for all $S \in \Sets$.
\item All the nodes labeled $+x$ in $\T$, for every $x$, are threaded in a list, with field $next$. The nodes labeled $-x$ are threaded in a similar list. This second threading must be consistent with the children ordering in $\T'$, that is, if nodes $u_1,u_2 \in \T$ labeled $-x$ are such that $down(u_1)$ is a sibling of $down(u_2)$ and precedes it, then $u_1$ precedes $u_2$ in the list. 
\item We use another pointer $last$ in $\T$, as follows: for each maximal sequence of nodes  $u_1,\ldots,u_k \in \T$ labeled $-x$ such that $down(u_1),\ldots,down(u_k) \in \T'$ are consecutive siblings, we set $last(u_1)=u_k$ (it may be that $k=1$ and thus $last(u_1)=u_1$).
\item For each $+x$ and each $-x$ present in $\T$, we store the head of the threaded list, with key $+x$ or $-x$ respectively, in a static dictionary.
\item We store a {\em mark} in every node of $\T'$, to be used at query time. The nodes are always unmarked, and are left unmarked again after each query. We also store a field $last$ in the nodes of $\T'$ to be used at query time.
\end{enumerate}

It is easy to see that the extra space incurred is in $O(\Delta(\Sets))$: we are adding $O(1)$ fields to the indel trees (which have $\Delta(\Sets)$ nodes), we are creating a new tree $\T'$ that is no larger than $\T$, and we are creating a dictionary using $O(1)$ space per distinct element $+x$ or $-x$ present in the indel tree nodes. \textcolor{black}{Figure~\ref{fig:support} illustrates the structures (some aspects will be made clear later).}

\begin{figure}[t]
\includegraphics[width=\textwidth]{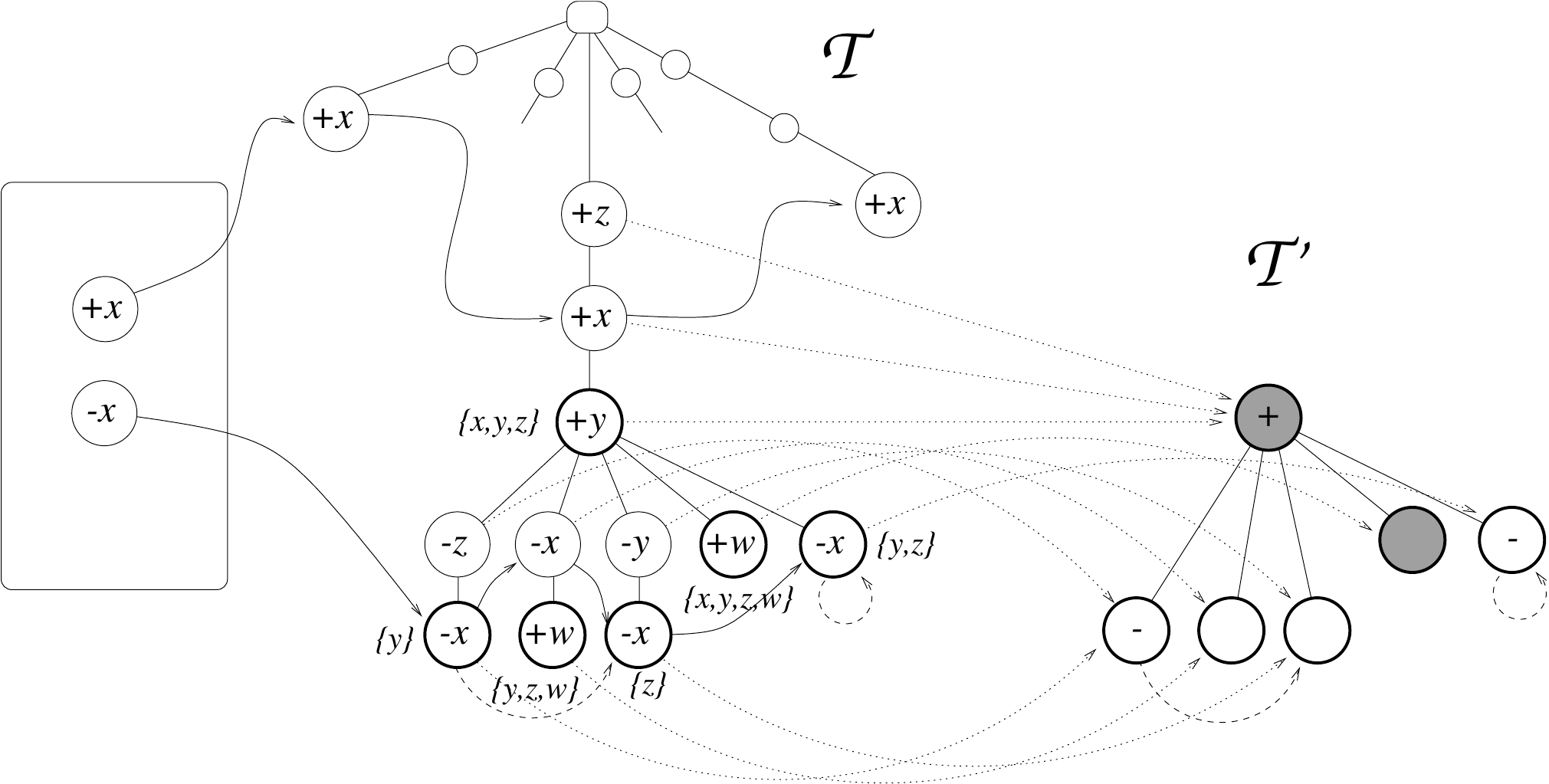}
\caption{Our data structures to compute support. On the left, the structure storing the heads of the threaded lists, which are shown (only for $+x$ and $-x$) with solid curved arrows. In the middle, a fragment of an indel tree $\T$, with its nodes representing some $v(S)$ in thick circles; the sets $S$ are written nearby them. On the right, the corresponding fragment of the tree $\T'$, with the $down$ pointers between trees shown as dotted arrows. To compute the support of $x$, the nodes of $\T'$ are marked with ``$+$'' and ``$-$'' (note how consecutive siblings that should be marked with ``$-$'' are not actually marked, thanks to the use of the $last$ pointers, which are shown with dashed arrows). The depth-first traversal starts at the gray node representing $\{x,y,z\}$ and visits the children, all of which are (explicitly or implicitly) marked with ``$-$'' except for the other gray node, which represents $\{x,y,z,w\}$. Those two gray nodes are then the only ones forming the support of $x$.}
\label{fig:support}
\end{figure}

\paragraph{Construction}
A depth-first traversal of $\T$ builds $\T'$ bottom-up, also setting the $down$ pointers. For the threaded lists of all nodes labeled $+x$ and $-x$, we maintain the distinct labels in a dynamic dictionary and traverse $\T$ again, appending each new label to its list or creating a new one in the balanced tree. To facilitate detecting contiguous siblings $down(u_i)$ in $\T'$ pointed from nodes $u_i$ with the same label $-x$, we traverse $\T$ in breadth-first order of $\T'$, processing all the nodes in the path from $v(S_1) \in \T$ to $v(S_2) \in \T$ when it is the turn to traverse the node $v'(S_2) \in \T'$ with parent $v'(S_1) \in \T'$. The nodes $u_i$ then appear contiguously in the threaded list of $-x$. Finally, we build the static dictionary from the dynamic one. Note that those dictionaries store $\min(|\U|,|\Delta(\Sets)|)$ keys. For the dynamic dictionary we use a linear-space predecessor data structure. We can leave that as the static dictionary, which will add $O(\log\log|\U|)$ to the query time, or use a perfect hash table for the static dictionary, which will add  expected time linear in $\min(|\U|,|\Delta(\Sets)|)$ to the construction.

\subsection{Query algorithm}

\noindent
Let us first find the support of $x$ that is encoded in $\T = \T_\emptyset$; we consider $\T_\U$ later. Moreover, let us first describe a simple non-optimal algorithm to compute $\supp(x,\Sets)$. First, if $+x$ is not in the static dictionary, then the answer is the empty set. Otherwise, traverse the threaded list of all the nodes $u$ labeled $+x$, and set a mark ``$+$'' on  $down(u)$. Further, if $-x$ is in the dictionary, traverse the threaded list of all the nodes $u$ labeled $-x$, and set a mark ``$-$'' on $down(u)$. It is easy to see that no node can receive both marks, but some may receive none. The key property is that, for every node $v'(S) \in \T'$, we must list $S$ iff the nearest marked ancestor of $v'(S)$ is marked ``$+$''.
We then traverse again the threaded list of $+x$, and from each node $u$ in the list, we start a depth-first traversal from $down(u)$, toward all of its descendants in $\T'$. This traversal stops when we reach a leaf or when we reach a node marked ``$-$''. After reporting all the sets represented by the traversed nodes not marked ``$-$'', we retraverse the threaded lists to remove the marks in $\T'$, leaving it ready for the next query.

To obtain optimal time $O(1+|\supp(x,\Sets)|)$, we must avoid reaching many consecutive siblings marked ``$-$'' in $\T'$. To do this, when we traverse the threaded list of $-x$, we do not use the pointer $next(u)$, but instead $next(last(u))$. This avoids marking with ``$-$'' the intermediate nodes of consecutive siblings in $\T'$. At the same time, we set $last(down(u)) = down(last(u))$ in $\T'$. When we traverse a subtree of a node marked ``$+$'' in $\T'$ and reach a node $v$ marked ``$-$'', we jump directly to its sibling $last(v)$, as we know that all the skipped siblings would be marked ``$-$'' in the simpler algorithm.

The only change needed to find the sets in $\T_\U$ is to assume that the root of $\T'_\U$ is marked ``$+$'' when we start any query (the space budget does not permit us explicitly storing $+x$ in the dictionary of $\T_\U$ for all the values of $x$).

\paragraph{Analysis} Since, for each node $u$ labeled $+x$ in $\T$ there is a distinct set $S \in \supp(x,\Sets)$ (i.e., $v'(S)=down(u)$), the traversal of the threaded list of $+x$ is within the time budget. Further, our depth-first traversal from nodes $down(u) \in \T'$ works $O(1)$ per distinct element of $\supp(x,\Sets)$ reported, or per node marked ``$-$'' it reaches. Because it traverses at most two consecutive children marked ``$-$'' (of the form $v$ and $last(v)$), the cost to process those two nodes can be charged to the preceding sibling of $v$ (which is unmarked), or to its parent if $v$ is a first child. This charges $O(1)$ to the parent and preceding siblings. As these nodes represent sets in $\supp(x,\Sets)$, the total charges are $O(|\supp(x,\Sets)|)$. It only remains to consider the cost of traversing the threaded list of nodes in $\T$ labeled $-x$: since we use $next(last(x))$ instead of $next(x)$, this is of the same order of the amount of nodes in $\T'$ marked with ``$-$'', which we have just analyzed.

\paragraph{Counting} Let $\Sets = \Sets_\emptyset \uplus \Sets_\U$ be the partition of $\Sets$ among the sets that are mentioned in $\T_\emptyset$ and in $\T_\U$, respectively. We can store $|\supp(x,\Sets_\emptyset)|$ within the dictionary we built for the indel tree $\T_\emptyset$, associated with the key $+x$, as there must exist one node labeled $+x$ in $\T_\emptyset$ if $x$ appears in some $S \in \Sets_\emptyset$ (that is, we know that $|\supp(x,\Sets_\emptyset)|=0$ if $+x$ is not in the dictionary). However, $x$ may occur in sets $S \in \Sets_\U$ without $+x$ occurring in $\T_\U$. Actually, $x$ appears in {\em all} the sets of $\Sets_\U$ unless $-x$ occurs in $\T_\U$. We then store $|\supp(x,\Sets_\U)|$ in the dictionary of $\T_\U$, associated with the key $-x$. If we do not find $-x$ in the dictionary, we then know that $|\supp(x,\Sets_\U)| = |\Sets_\U|$. Overall, we can compute $|\supp(x,\Sets)| = |\supp(x,\Sets_\emptyset)| + |\supp(x,\Sets_\U)|$ in constant time.

\begin{theorem} \label{thm:supp}
A set of sets $\Sets$ can be represented in $O(\Delta(\Sets))$ space so that $\supp(x,\Sets)$ can be computed in optimal $O(1+|\supp(x,\Sets)|)$ time, and $|\supp(x,\Sets)|$ can be computed in constant time. The structures can be built from the corresponding indel trees of $\Sets$ in time $O(|\Delta(\Sets)|\log\log|\U|)$, plus $O(\min(|\U|,|\Delta(\Sets)|))$ expected time. We can avoid this last expected time, computing $\supp(x,\Sets)$ in time  $O(\log\log|\U|+|\supp(x,\Sets)|)$ and $|\supp(x,\Sets)|$ in time $O(\log\log|\U|)$.
\end{theorem}

\subsection{Enumerating the support}

\noindent
A more challenging problem is to provide an {\em enumerator} for $\supp(x,\Sets)$, that is, a data structure that, given $x$, spends some initial time (that must be $o(1+|\supp(x,\Sets)|)$ or Theorem~\ref{thm:supp} renders it trivial), and then can output any new element of $\supp(x,\Sets)$ within some delay (which should be small too). For example, using consecutive access operations we can enumerate any set $S$ with constant initialization time and $O(\log|\U|)$ delay, within $O(\Delta(\Sets))$ space. Theorem~\ref{thm:supp} does not immediately provide an enumerator: while it would be trivial to traverse the list of $+x$ and launch depth-first traversals from each new node instead of marking them ``$+$'', we must traverse the list of $-x$ before launching any depth-first search. We must then find a way to output some elements as we traverse the list of $-x$ and set the ``$-$'' marks. Then we can continue with the $+x$ list and output all the elements we can reach from them.

We will assume some additional data associated with nodes of $\T'$: the children of a node $v \in \T'$ have $parent$ pointers and are arranged in a doubly linked list, so the preceding and following sibling is found in constant time. There will be fields $prev$ and $next$ in nodes of $\T'$ linking maximal first and last siblings marked with ``$-$''. Since we will not have time to initialize fields (or to restore them after the query), we use initializable arrays \cite[Sec.\ III.8.1]{Meh84} for all those fields that start with a fixed value at the beginning of queries (like marks in $\T'$).

Let us define an operation that extends a list $v,\ldots,v'$ of nodes we are marking with ``$-$'' to other neighboring siblings already marked ``$-$'', so as to maintain maximality of those marked siblings. We first set $next(v)=v'$ and $prev(v')=v$. We then try to concatenate $v,\ldots,v'$ with a possible preceding and following list of nodes marked ``$-$'': let $w$ be the sibling preceding $v$. If it exists and is marked ``$-$'', we set $next(prev(w)) = v'$ and $v = prev(v') = prev(w)$; we proceed analogously with the sibling following $v'$.

We then traverse the list $-x$ as follows. For each new node $u$, we find $v=down(u)$ and $v' = last(v) = down(last(u))$ as before, and extend $v,\ldots,v'$. We must now produce some output for the enumerator. If there is a previous sibling $w$ of $v$, we report it, mark it ``$-$'', and extend $w,\ldots,w$ (note $w$ cannot appear in the list of $+x$). If there is no sibling preceding $v$, we try to do the same with the next sibling $w$ of $v'$. If none exist, then all the children of $parent(v)$ are marked ``$-$''. We then report $w=parent(v)$ instead (which cannot be in the list of $-x$).

We must now ensure that $w=parent(v)$ will not be reported. If $w$ is not $down(u)$ for some node in the list of $+x$, then we can simply mark $w$ as ``$-$'' and and extend the list $w,\ldots,w$, which gradually extends upwards the list of nodes descending from some $+x$ that are enumerated as we process the list of $-x$. However, if $w=down(u)$ for some $u$ in the list of $+x$, we must instead remove $u$ from that list: we have exhausted all the nodes reachable from it and cannot spend time on it when we traverse the $+x$ list. For this purpose, we can maintain for each distinct $+x$, another static dictionary of those nodes $u \in \T$, with key $down(u) \in \T'$, so as to look for $w$ and remove $u$ from the list of $+x$. Once we exhaust the list of $-x$, we traverse (what remains of) the list of $+x$ as explained, reporting all elements until the leaves or until reaching ``$-$'' marks.

The removal of $u$ from the list of $+x$ is only local to the execution of the current query. We maintain fields $succ$ in the list, which if initialized point to the next non-eliminated element, and otherwise are understood to point to the next element in the list. So we temporarily remove $u$ by taking the previous element in the list, $pred$, and setting $succ(pred(u)) = succ(u)$. Note our delay is $O(1)$ if we build perfect hash tables, otherwise it is $O(\log\log |\Sets|)$.

\begin{theorem} \label{thm:enum}
A set of sets $\Sets$ can be represented in $O(\Delta(\Sets))$ space so that $\supp(x,\Sets)$ can be enumerated with constant initialization time and delay. The additional structures, on top of those of Theorem~\ref{thm:supp}, can be built in $O(\Delta(\Sets))$ expected time. We can convert that construction time into worst-case by increasing the delay to $O(\log\log |\Sets|)$.
\end{theorem}

\subsection{The inverted index}

\noindent
We can regard $\Sets$ as a binary matrix $M$ of $|\Sets|$ rows and $|\U|$ columns, where $M[S,x]=1$ iff $x \in S$. Let us then define our compressibility measure over those matrices, $\Delta(M) = \Delta(\Sets)$. Our operations described in the Introduction then give us sophisticated access operations on the rows of this matrix, that is, within some fixed set $S$. Retrieving the support of $x$, instead, corresponds to recovering the 1s in column $x$. With this view, it is natural to consider the transposed matrix, $M^\mathrm{t}$, where our operations access, rank, and predecessor/successor give us sophisticated access to the support of elements, and the support operation defined in this section allows retrieving any set $S$ in optimal time (while our access operation on $M$ retrieves it in time $O(|S|\log|\U|)$). Note that the operation membership gives the same functionality in both representations.

The transposed matrix corresponds to representing another set, $\Sets^\mathrm{t}$: for each $x \in \U$, $\Sets^\mathrm{t}$ contains the set $S_x = \{S \in \Sets, x \in S\}$ (in Information Retrieval applications, where the sets $S$ are documents and $x$ are the words contained in it, $\Sets^\mathrm{t}$ is the {\em inverted index} of $\Sets$: for each word $x$, it tells the documents where it occurs). We show, however, that there is no relation between $\Delta(\Sets^\mathrm{t})$ and $\Delta(\Sets)$, and therefore we cannot maintain both matrices within space $O(\Delta(\Sets))$.

\begin{theorem}
There is an infinite family of matrices $M$ where $\Delta(M) {=} o(\Delta(M^\mathrm{t}))$, and vice versa.
\end{theorem}
\begin{proof}
The rough idea is as follows: let $M$ be a $2n \times 2n$ matrix for some perfect square $n$, and let $r = \sqrt{n}$.
Let $M_i$, for $i=1,\ldots,r$, be nonzero random binary $r \times r$ matrices, and let $I$ be the $n\times n$ identity matrix (which we will use to ensure that all rows and columns are different). Then, $M$ is defined as follows, each quadrant being of size $n \times n$:
\[
M ~=~~ \left(
\begin{array}{cccc|cccc}
 M_1 & M_2 & \cdots & M_r & \\
 M_1 & M_2 & \cdots & M_r & \\
     & \vdots & & & \\
 M_1 & M_2 & \cdots & M_r & & \text{\smash{\raisebox{1.5em}{\huge $~~~I~~~$}}} & \\
\hline 
  & & & & \\
  & & & & \\
  & & & & \\
  & \text{\smash{\raisebox{1.5em}{\huge $~I$}}} & & & & \text{\smash{\raisebox{1.5em}{\huge $~~~I~~~~$}}} \\
\end{array}
\right)
\]
It then follows that $\Delta(M) = O(n^{3/2})$: the first $r$ rows require $O(rn)$ space, but then the next $n-r$ rows require just $O(1)$ space each, as there exists a previous row differing in $O(1)$ cells. The last $n$ rows also require $O(1)$ space each. Instead, $\Delta(M^\mathrm{t})$ is $\Theta(n^2)$, because the first $n$ columns have $n$ cells that are random and unrelated with any other column (though they are internally repetitive, this cannot be exploited column-wise).

To complete the proof, let us define a deterministic way to build the matrices $M_i$. The only important aspect is to ensure the column-wise lower bound. In particular, we must ensure that there is a constant $c$ such that every pair of distinct columns in the concatenation $M_1 : \cdots : M_r$ differ in at least $r/c$ bits. This ensures that $\Delta(M^\mathrm{t}) \ge n \cdot r \cdot r/c = \Omega(n^2)$.

The problem is equivalent to finding $r$-bit vectors whose pairwise Hamming distance is at least $d=r/c$. The amount of such bit vectors is known as $A(r,d)$ in coding theory \cite[Ch.~17]{errorcodes}. Consider the following constructive argument: choose any vector, and eliminate all those at Hamming distance at most $d$ from it. From the remaining vectors, choose another, eliminate its neighbors, and so on. Since we eliminate at most $\sum_{j=0}^d {r \choose j}$ vectors for each one we choose (``at most'' because some can be eliminated already), and because $\sum_{j=0}^d {r \choose j} \le 2^{r H(d/r)}$, where $H$ is the binary entropy function, the bound of Gilbert-Varshamov \cite{gilbert1952comparison,varshamov1957estimate} follows:
\[ A(r,d) ~\ge~ 2^r / \textstyle\sum_{j=0}^d {r \choose j} ~\ge~ 2^{r-rH(d/r)} ~=~ 2^{r(1-H(1/c))}. \]
This bound is of the form $2^{\alpha r}$ for a positive constant $\alpha$ for any $c \ge 3$. Since $2^{\alpha r} \subset \Omega(r^2) = \Omega(n)$ for any positive $\alpha$, there will be sufficient columns from where to choose the $n$ we need to form $M_1,\ldots,M_r$.

The ``vice versa'' part is symmetric, since $(M^\mathrm{t})^\mathrm{t} = M$.
\end{proof}

We note that, depending on which is smaller, we might prefer to store $M$ or $M^\mathrm{t}$, since both offer similar functionality for membership queries. The support for the other operations, however, is not symmetric.

\section{\textcolor{black}{Stronger Compressibility Models}}

\noindent
As it is evident from Figure~\ref{fig:ientropy}, one could reduce the number of nodes in the insertion tree by factoring paths; say, the paths for $S_6$ and $S_5/S_8$ could start with a shared part with the elements $\{3,5,6\}$. This corresponds to adding such set to $\Sets$. We can define a more refined form of $\IEntropy(\Sets)$ (and of $\DEntropy(\Sets)$) where one can optimally add such convenient subsets, or equivalently, allow branching internal nodes that represent no set in $\Sets$. Figure~\ref{fig:ientropyp} (left) shows a tree with $17 < \IEntropy(\Sets) = 20$ nodes (plus the root). All our machinery to query the insertion-tree-based, or the indel-tree-based, representations would work verbatim.

\begin{figure}[t]
\includegraphics[width=0.9\textwidth]{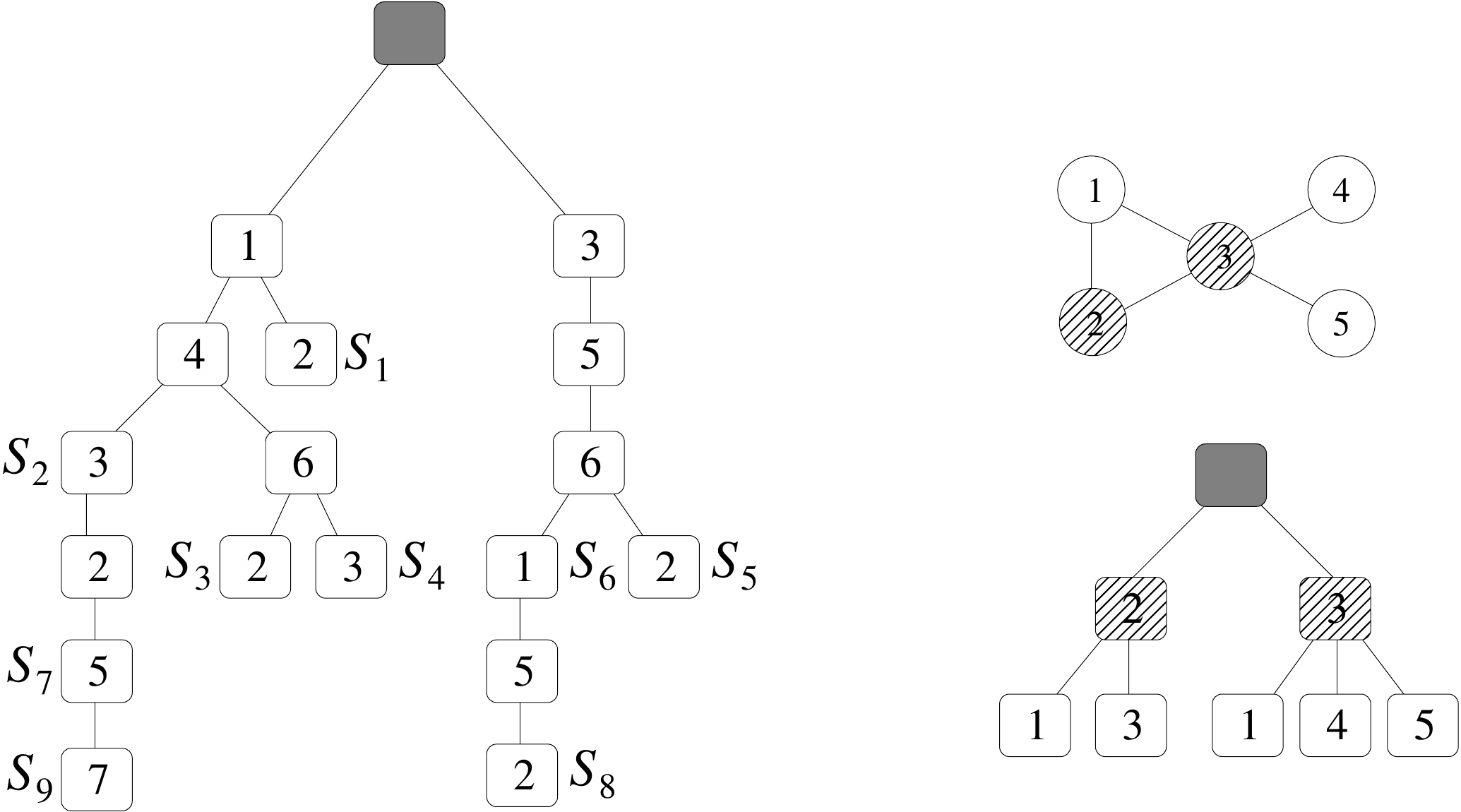}
\caption{On the left, a tree with $\IEntropy'(\Sets)=17$ nodes, plus root, for the set of Figure~\ref{fig:ientropy}. On the right, a reduction from Vertex Cover to finding $\IEntropy'(E)$ on the set $E$ of graph edges; a smallest vertex cover is $V=\{2,3\}$.}
\label{fig:ientropyp}
\end{figure}

An even more general setup allows more powerful operations to build the sets, other than just adding or deleting elements from another set. In particular, we could allow complements, unions, and intersections of other sets. In this case, it is not obvious how to extend our techniques to efficiently query the sets $S$ from the representation.

An advantage of the symdiff complexity is that the optimal representation can be built efficiently (in time polynomial in $n=\sum_{S \in \Sets} |S|$). In the rest of this section we consider the construction problem for the extended versions we have described. As we show, the extensions either do not asymptotically improve the symdiff complexity, or they are NP-hard to build optimally, and even to approximate within constant factors.

\subsection{Creating convenient intermediate sets}

\noindent
Let us first consider the case of the insertion complexity. If we allow adding new sets that can act as branching nodes in the insertion tree, then our problem is: given a set $\Sets$, find the smallest tree with nodes labeled with elements of $\U$ where, for each $S \in \Sets$, there is a tree node whose path from the root contains exactly the labels in $S$. Let us call $\IEntropy'(\Sets) \le \IEntropy(\Sets)$ the number of nodes (not counting the root) of this smallest tree.

Consider the special case where all the sets $S \in \Sets$ have size $|S|=2$. Since we can build sets only with insertions, every $S$ must be built from the empty set, so $\IEntropy(\Sets)=2|\Sets|$. For $\IEntropy'(\Sets)$, an optimal tree would have height 2, with exactly $|\Sets|$ leaves at depth 2 and a variable-size set $V$ of internal nodes (at depth 1). The tree has then $1+|V|+|\Sets|$ nodes and we aim to minimize $|V|$. If we regard $\Sets$ as a set of undirected edges, then minimizing the amount of internal nodes amounts to finding a smallest set $V \subseteq \U$ of vertices in the graph $G(\U,\Sets)$ such that every edge in $\Sets$ contains one of its extremes in $V$ (see Figure~\ref{fig:ientropyp} (right)). 

We can then reduce from Vertex Cover: given a graph with edge set $E$, we consider $E$ as a collection of pairs of vertices, so all sets have cardinality 2. We can find a smallest vertex cover for $E$ by reducing it to the problem of finding the smallest tree for $\Sets = E$. Therefore, even this particular case of our problem is NP-hard.

Note that, when $|S|=2$, we also have $\Delta(\Sets)=2|\Sets|$ because we need at least two operations to build a set from any other. Let us now extend the symdiff complexity: $\Delta'(\Sets)$ will be the size of the smallest pair of indel trees representing $\Sets$, where we permit branching internal nodes that do not represent any set in $\Sets$. In the case where $|S|=2$, however, the optimal solution is as for $\IEntropy'(\Sets)$: (i) it is optimal to use only the tree rooted at $\emptyset$; (ii) in an optimal tree, the nodes at depth 1 and 2 must be of the form $+x$; (iii) to form another set $S$ from a depth-2 node, we should at least delete an element and add another, which costs 2, and therefore it is no better than building $S$ from the root. Therefore, $\Delta'(\Sets)=\IEntropy'(\Sets)$ in this case, and thus the problem is NP-hard as well.

If we consider only the node $\emptyset$ or $\U$, then the problem is exactly the Undirected Steiner Tree problem on the Boolean lattice, with terminals $S \in \Sets$. Even with our metric distance, the problem is APX-hard \cite{chlebik2008steiner}, but a $\approx1.4$-approximation exists \cite{byrka2013steiner}. On the other hand, it is well-known that introducing Steiner points can reduce the MST weight by at most a half \cite{gilbert1968steiner}, so the representation cost is still $\Theta(\Delta(\Sets))$. If we consider the insertion complexity, $\IEntropy(\Sets)$, the problem corresponds to the Directed Steiner Tree problem, which is harder to approximate.
Polynomial-time algorithms \cite{Charikar99} achieve approximation ratio $O(s^\epsilon/\epsilon^3)$ for any fixed $\epsilon>0$, in time $O(s^{4/\epsilon}m^{2/\epsilon})$, where $s = |\Sets|$ and $m = \sum_{S \in \Sets} 2^{|S|}$ ($m$ might be polynomial or not on our input size, $n = \sum_{S \in \Sets} |S|$). The best possible approximation ratio within quasi-polynomial time complexity (i.e., $O(s^{\mathrm{polylog}\,m})$) is known to be $\Theta(\log^2 s / \log\log s)$ \cite{Grandoni19}.

\subsection{Allowing set complementation}

\noindent
Our measure $\Delta(\Sets)$ already supports set complementation. Concretely, given a representation of $\Sets$, we can create an $O(\Delta(\Sets))$-size representation of $\Sets \cup \Sets'$, where $\Sets' = \{ \U \setminus S,~S \in \Sets \}$ contains the complements of all the sets in $\Sets$: Assume we have the \textcolor{black}{symdiff forest} for $\Sets$, and let $\T_\emptyset$ and $\T_\U$ be the resulting indel trees rooted at $\emptyset$ and $\U$, respectively. We can then represent $\Sets'$ with two indel trees: $\T'_\emptyset$, rooted at $\emptyset$, is isomorphic to $\T_\U$, and $\T'_\U$, rooted at $\U$, is isomorphic to $\T_\emptyset$. Every label $+x$ ($-x$) in $\T_\emptyset$ or $\T_\U$ becomes a label $-x$ ($+x$) in $\T'_\U$ or $\T'_\emptyset$, respectively. It is easy to see that the indel trees $\T'_\emptyset$ and $\T'_\U$ represent, precisely, the complements of the sets represented by $\T_\emptyset$ and $\T_\U$, at the same nodes. Therefore, the unions $\T_\emptyset \cup \T'_\emptyset$ (with root $\emptyset$) and $\T_\U \cup \T'_\U$ (with root $\U$) represent $\Sets \cup \Sets'$ within space $O(\Delta(\Sets))$.

Such an expanded representation of $\Sets$ allows us taking the complement of any set $S$, even at zero cost, by switching to the node $\U\setminus S$, which exists in the other tree. Therefore set complementation does not add to the power of the symdiff compressibility in asymptotic terms; it is already supported.

\subsection{Allowing unions, intersections, and complements}

\noindent
If, in addition to adding or removing single elements (i.e., labels $+x$ and $-x$), and complementing sets, we permit intersections and unions of other sets, the problem is called multitarget circuit complexity: build the multiple ``targets'' $S \in \Sets$ from the ``generators'' $x \in \U$ using the least number of set operations (adding a single element, $+x$, is a particular case of union, and removing a single element, $-x$, is achieved by complementing, adding $x$, and complementing again, at cost $1$). Finding the least number of operations to produce the targets is known to be NP-hard, and hard to approximate within any factor $|\U|^{1-\epsilon}$, for any constant $\epsilon>0$ \cite{Ilango20}.

\section{Applications and Previous Work} \label{sec:motivation}

\noindent
There are various applications where it is natural to encode sets of sets by their symmetric differences. An example is inverted indices in natural language text collections, which store the sets of documents where each word occurs. The phenomenon that semantically correlated words tend to appear in about the same sets of documents has been observed long ago \cite{SL68}\cite[Ch.~3]{BYRN11}, and for example is used to detect word associations in NLP. Another example is the adjacency lists of web graphs and social networks, where the symmetric differences between such lists have been used to compress the graphs with much success \cite{BV04,BSV09,BRSV11}. This idea has been extended to arbitrary Boolean matrices (which in particular can be the adjacency matrices of graphs).

Elgohary et al.'s papers ~\cite{EBHRR18,EBHRR19} have contributed to a burst of research~\cite{ABBK20,AMB+25,AGN24,BB23,DCA22,FMGKNST22,MFMF23,TBBDFM25} on compressing matrices and manipulating them in compressed form.  The compression methods used by Elgohary et al.\ and the researchers that followed them usually treat the rows of the matrices as sequences, however --- possibly after re-ordering the columns and/or rows --- and thus have difficulty taking advantage of repetitions of a pair of values in two columns when they are separated by columns whose contents vary.

In contrast, over thirty years ago Bookstein and Klein~\cite{BK91} proposed compressing collections of bitvectors as collections of sets: we can express one set in terms of a similar one by recording their symmetric difference, and we can find the best way to express a collection this way by considering the complete graph whose vertices are the sets and whose edges are weighted by the cardinalities of the symmetric differences of their endpoints, and building a spanning forest of rooted trees that minimizes the sum of the cardinalities of the sets at the roots (which we store explicitly) and the total weights of the edges in the forest.

Bookstein and Klein's idea was reinvented several times~\cite{AM01,BBHKV98,BL01,BCCFM00}, but it seems Bj\"orklund and Lingas~\cite{BL01} were the first to observe that we can multiply matrices in time bounded in terms of the weights of their forests (the cardinalities of the sets at the roots and the total weights of the edges).  To see why, consider that if we have already multiplied a row $\vec{u}$ of binary matrix {\bf A} by a column $\vec{v}$ of binary matrix {\bf B} and we know the symmetric difference of $\vec{u}$ and another row $\vec{u}'$ of {\bf A}, viewed as sets, then we can compute the product of $\vec{u}'$ and $\vec{v}$ by multiplying the elements in that symmetric difference by the corresponding elements in $\vec{v}^\top$ and adjusting the product $\vec{u} \cdot \vec{v}^\top$ appropriately.  This observation easily generalizes also to non-binary matrices.

In the case of adjacency matrices of webgraphs with the rows sorted by URL, neighbouring rows are more likely to be similar.  Boldi and Vigna~\cite{BV04,BSV09,BRSV11} took advantage of this tendency and considered only edges between rows close in that ordering, in order to speed up the construction of spanning forests that usually still have low weight in practice for webgraphs.  They limited the possibilities of compression in order to provide a reasonable extraction time for the adjacency lists (i.e., the nonempty cells in a row).  Grabowski and Bieniecki~\cite{GB14} optimized their heuristic and Francisco et al.~\cite{FGKLN22} observed that the compression could be used to speed up matrix-vector multiplication.  They were all apparently aware of some earlier work but not Bookstein and Klein's nor Bj\"orklund and Lingas's papers in particular. 

In the case of coloured de Bruijn graphs (CDBGs) for bioinformatics, the sets of colours of neighbouring vertices are more likely to be similar.  Almodaresi et al.~\cite{APFJP20} took advantage of this tendency when compressing CDBGs for use in their tool Mantis~\cite{PABFJP18} by considering only edges between sets of colours associated with neighbouring vertices, and encoding the resulting color set on the minimum spanning tree of the induced graph.  \travis{Campanelli et al.~\cite{CPFP26} took this idea further by considering unitigs in the CDBGs (paths of vertices with in- and out-degree 1, which become single edges in the compacted CDBG) and compressing those sets of colours.  Levallois et al.~\cite{LSLDPPP26} gave an approximate solution using minimizers~\cite{SWA03,RHHMY04} and integer compression.}

Alanko et al.~\cite{ABGNP25} recently proposed compressing a collection of sets by finding, for each set $S$ in the collection, the smallest superset $S'$ of it in the collection and making $S'$ the parent of $S$ in a tree (this is close to a notion of ``deletion complexity'' in our setting).  They path-compress this tree to make its height logarithmic in the size of the universe. 
They encode $S$ as a bitvector indicating which elements of $S'$ 
are in $S$, which allows them to support queries such as predecessor and successor on the sets in logarithmic time.  They tested this idea on CDBGs in the tool Themisto~\cite{AVMP23} with 16 thousand bacterial genomes and found it improved compression compared to Themisto's default (0.18 bits per element instead of 0.32 bits), although they did not report query times.

\travis{Somewhat reminiscent of Campanelli et al.'s work but for variation graphs instead of CDBGs, Bonizzoni et al.~\cite{BCDGGR26} very recently introduced the problem of placing the minimum number of coloured pebbles on a directed acyclic graph (DAG) such that each of several coloured paths through the DAG can be reconstructed by the edges pebbled with pebbles of their colours.  They showed this problem can be solved in polynomial time but becomes NP-hard when every pebbled edge must be pebbled with the colour of every path crossing it.}

\section{Conclusion and Further Work}

\noindent
We have introduced a measure, the symdiff complexity $\DEntropy(\Sets)$, of the space needed to optimally represent a set $\Sets$ of sets by indicating which elements from each set differ from those of some other set. Such representation has been used already in various applications, which demonstrates its practical interest. Our contributions are (1) formalizing the measure; (2) giving an improved algorithm to compute $\DEntropy(\Sets)$ based on Prim's MST construction with edge costs computed incrementally; (3) designing an $O(\Delta(\Sets))$-space representation, based on tree extraction, that supports the most fundamental operations on the sets of $\Sets$ (membership, access, rank, predecessor, successor\textcolor{black}{, support}) in time at most $O(\log u)$ \textcolor{black}{per output}, where $u$ is the universe size.

Our adaptation of Prim (2) is of independent interest: it assumes that (integer) edge costs $c$ can be computed incrementally: at any moment we know the cost is at least $c$ and, in $O(1)$ time, we can determine whether it is $c$ or greater. If the maximum cost of an MST edge is $\ell$, then the total cost of our Prim's variant is the sum, for every edge of cost $c$, of $\min(c,\ell)$. 

Another result of independent interest, in (3), is an $O(\log u)$-time extension of the known algorithm to find the $i$th smallest label among the ancestors of a labeled tree node. In the extension, some nodes can be marked as ``deleting'' a label that occurs upwards\textcolor{black}{; we show that this generalizes the segment rank and select problem~\cite{bgpt25}}.
We leave for future work to achieve $O(\log_\omega u)$ time for this extension, which is the cost of the basic version with no deletion marks. This would make the times of the five basic queries within $O(\DEntropy(\Sets))$ space asymptotically the same as those we obtain within the larger space $O(\IEntropy(\Sets))$.

\textcolor{black}{We have also obtained succinct space for both measures, $|\IEntropy(\Sets)|\lg u (1+o(1))$ and $|\Delta(\Sets)|\lg u (1+o(1))$ bits, instead of just $O(|\IEntropy(\Sets)|)$ or $O(\DEntropy(\Sets))$ words, respectively, though query times raise to $O(\log^2 u)$ in the second case. Achieving succinctness without sacrificing query times for symdiff compressibility is another future work challenge.}

\textcolor{black}{It is also interesting to 
further explore the relation between rows and columns,
for example, computing the contents of a set $S \in \Sets$ in optimal time $O(|S|)$ (i.e., extract all the 1s in a row) just as we did for columns (i.e., computing the support $\supp(x,\Sets)$). This is trivial within insertion compressibility space, O$(\IEntropy(\Sets))$, but not within $O(\Delta(\Sets))$ space, where we can only compute $S$ in time $O(|S|\log|\U|)$ using our access operation.} 

\textcolor{black}{Finally, we explored some stronger models of compressibility beyond symdiff, but found them to be NP-hard to compute in general. It would be interesting to study other extensions of our model that are relevant in practice and can still be efficiently built from the data.}

\anon{
\section*{Acknowledgements}

\noindent
Thanks to Arturo Merino, Rob Patro, and Bernardo Subercaseaux for helpful discussions.
}

\bibliographystyle{splncs04}
\bibliography{paper}
\end{document}